\documentclass[11pt,a4paper]{article}
\usepackage{amssymb}
\usepackage{amsmath}
\usepackage{amsfonts}
\usepackage{bbm}
\usepackage{amsthm}
\usepackage{mathrsfs}
\usepackage{hyperref}
\usepackage{color}
\usepackage[margin=2.41cm]{geometry}
\usepackage[all,cmtip]{xy}
\usepackage[utf8]{inputenc}
\usepackage{graphicx}
\usepackage{varwidth}
\usepackage{comment}
\usepackage{slashed}
\usepackage{upgreek}
\usepackage{rotating}
\usepackage{tikz-cd}

\usepackage{xfrac}

\usepackage[shortlabels]{enumitem}

\usepackage{tikz}
\usetikzlibrary{arrows.meta,bending,decorations.markings}
\tikzset{
    arc arrow/.style args={%
    to pos #1 with length #2}{
    decoration={
        markings,
         mark=at position 0 with {\pgfextra{%
         \pgfmathsetmacro{\tmpArrowTime}{#2/(\pgfdecoratedpathlength)}
         \xdef\tmpArrowTime{\tmpArrowTime}}},
        mark=at position {#1-\tmpArrowTime} with {\coordinate(@1);},
        mark=at position {#1-2*\tmpArrowTime/3} with {\coordinate(@2);},
        mark=at position {#1-\tmpArrowTime/3} with {\coordinate(@3);},
        mark=at position {#1} with {\coordinate(@4);
        \draw[-{Triangle[length=#2,bend]}]       
        (@1) .. controls (@2) and (@3) .. (@4);},
        },
     postaction=decorate,
     },
fermion arc arrow/.style={arc arrow=to pos #1 with length 2.5mm},
Vertex/.style={fill,circle,inner sep=1.5pt},
insert vertex/.style={decoration={
        markings,
         mark=at position #1 with {\node[Vertex]{};},
        },
     postaction=decorate}     
}

\usetikzlibrary{decorations.pathmorphing}
\tikzset{snake it/.style={decorate,decoration={zigzag, segment length=5, amplitude=1.5}}}

\definecolor{darkred}{rgb}{0.8,0.1,0.1}
\hypersetup{
     colorlinks=false,         
     linkcolor=darkred,
     citecolor=blue,
}

\theoremstyle{plain}
\newtheorem{theo}{Theorem}[section]

\newtheorem{propo}[theo]{Proposition}

\theoremstyle{definition}
\newtheorem{defi}[theo]{Definition}

\newtheorem{exx}[theo]{Example}
\newenvironment{ex}
  {\pushQED{\qed}\exx}
  {\popQED\endexx}

\newtheorem{remm}[theo]{Remark}
\newenvironment{rem}
  {\pushQED{\qed}\remm}
  {\popQED\endremm}

\numberwithin{equation}{section}

\def\nn{\nonumber}

\def\bbR{\mathbb{R}}
\def\bbC{\mathbb{C}}

\def\bbZ{\mathbb{Z}}

\def\ii{{\,{\rm i}\,}}

\def\dR{\mathrm{dR}}

\def\End{\mathrm{End}}

\def\Aut{\mathrm{Aut}}
\def\Sym{\mathrm{Sym}}
\def\Obs{\mathrm{Obs}}

\def\cl{\mathrm{cl}}
\def\qu{\mathrm{qu}}
\def\fr{\mathrm{free}}
\def\intt{\mathrm{int}}
\def\BV{\mathrm{BV}}

\def\id{\mathrm{id}}

\def\dd{\mathrm{d}}
\def\oone{\mathbbm{1}}

\def\Tr{\mathrm{Tr}}
\def\Mat{\mathrm{Mat}}

\def\AA{\mathcal{A}}
\def\HH{\mathcal{H}}
\def\DD{\mathcal{D}}
\def\GG{\mathcal{G}}
\def\LL{\mathcal{L}}

\def\su{\mathfrak{su}}
\def\g{\mathfrak{g}}

\def\sk{\vspace{1mm}}

\newcommand{\pair}[2]{\langle #1 , #2\rangle}
\newcommand{\cyc}[2]{\langle\!\langle #1 , #2\rangle\!\rangle}

\makeatletter
\let\@fnsymbol\@alph
\makeatother

\title{
BV quantization of dynamical fuzzy spectral triples
}

\author{
James Gaunt$^{1,2,a}$, Hans Nguyen$^{1,b}$\ and\ Alexander Schenkel$^{1,c}$\vspace{4mm}\\
{\small ${}^1$ School of Mathematical Sciences, University of Nottingham,}\\
{\small University Park, Nottingham NG7 2RD, UK.}\vspace{2mm}\\
{\small ${}^2$ Department of Mathematics, Heriot-Watt University,}\\
{\small Colin Maclaurin Building, Riccarton, Edinburgh EH14 4AS, UK.}\vspace{4mm}\\
{\small \begin{tabular}{ll}
Email: &  ${}^a$~\texttt{j.gaunt@hw.ac.uk}\\
& ${}^b$~\texttt{hans.nguyen@nottingham.ac.uk}\\
& ${}^c$~\texttt{alexander.schenkel@nottingham.ac.uk}\vspace{2mm}
\end{tabular}
}
}

\date{September 2022}

\begin{document}

\maketitle


\begin{abstract}
\noindent This paper provides a systematic study of gauge symmetries in the dynamical fuzzy spectral triple models for quantum gravity that have been proposed by Barrett and collaborators. We develop both the classical and the perturbative quantum BV formalism for these models, which in particular leads to an explicit homological construction of the perturbative quantum correlation functions. We show that the relevance of ghost and antifield contributions to such correlation functions depends strongly on the background Dirac operator $D_0$ around which one perturbs, and in particular on the amount of gauge symmetry that it breaks. This will be illustrated by studying quantum perturbations around 1.) the gauge-invariant zero Dirac operator $D_0=0$ in a general $(p,q)$-model, and 2.) a simple example of a non-trivial $D_0$ in the quartic $(0,1)$-model.
\end{abstract}


\renewcommand{\baselinestretch}{0.8}\normalsize
\tableofcontents
\renewcommand{\baselinestretch}{1.0}\normalsize



\section{Introduction and summary}
Noncommutative geometry provides a powerful and versatile mathematical framework
that allows one to incorporate quantum effects into the small-scale structure
of spacetime. Throughout the past two decades, we have witnessed the development
of different, but related, approaches to the problem of describing
generalizations of Riemannian geometry, and hence of gravity, to noncommutative spaces. 
The main difference in these approaches lies in how they attempt to 
encode the geometry/gravitational field on a noncommutative space. 
The most common options are metric approaches
\cite{AschieriMetric,BeggsMajid}, vielbein approaches \cite{AschieriCastellani}
and more radical approaches such as Connes' spectral triples \cite{Connes}.
\sk

This paper is a contribution to the spectral triple approach to noncommutative (quantum) gravity.
More precisely, we shall work within the framework of {\em fuzzy spectral triples} developed by
Barrett in \cite{Barrett1}, which can be thought of as finite dimensional approximations of Euclidean spacetime. 
Very informally speaking, a spectral triple $(\AA,\HH,D)$ consists of a possibly noncommutative
algebra $\AA$ that is represented on a Hilbert space $\HH$, together with a Dirac operator $D$ on $\HH$.
(Further data is required for a real spectral triple,  namely a chirality operator $\Gamma$ and a real structure $J$, 
see Section \ref{sec:spectral} for more details.) The physical interpretation
is that $\AA$ is the ``algebra of functions'' on a noncommutative space, $\HH$ is the space of spinor fields
and the Dirac operator $D$ encodes the noncommutative Riemannian geometry. Fuzzy spectral triples
are particular examples of real spectral triples with $\AA = \Mat_{N}^{}(\bbC)$ a finite-dimensional
matrix algebra and $\HH = \AA\otimes V$ determined by a module $V$ over some Clifford algebra $\mathsf{Cl}_{p,q}$. 
See \cite{Barrett1} for the description
of the fuzzy sphere in this framework and \cite{BarrettGaunt} for the fuzzy tori.
\sk

Fuzzy spectral triples can be used to develop and study (toy-)models for
quantum gravity theories on noncommutative spaces \cite{Barrett3}. The basic idea is to consider
the Dirac operator $D$ as a dynamical variable that gets quantized
through performing path-integrals over the space $\DD$ of Dirac operators.
More precisely, given a suitable action $S: \DD\to \bbR$ on the space of
Dirac operators, one defines the partition function of such quantum gravity 
model as the integral
\begin{subequations}\label{eqn:pathintegral}
\begin{flalign}
Z\,:=\,\int_{\DD}e^{-S(D)}\,\dd D
\end{flalign}
and the expectation value of an observable $\mathcal{O}:\DD\to\bbC$ by
\begin{flalign}
\langle \mathcal{O}\rangle\,:=\frac{1}{Z}\,\int_{\DD} \mathcal{O}(D)\,e^{-S(D)}\,\dd D\quad.
\end{flalign}
\end{subequations}
It is important to stress that, for fuzzy spectral triples, the space of Dirac operators 
$\DD$ is finite-dimensional. Hence, such integrals exist rigorously, subject to suitable
conditions on the action $S$. Furthermore, using the classification of finite spectral triples 
\cite{Krajewski}, together with the explicit characterization of Dirac 
operators in terms of (anti-)Hermitian matrices \cite{Barrett1},
the path-integrals in \eqref{eqn:pathintegral} may be reformulated
in terms random multi-matrix models. The original study in \cite{Barrett3}
was through numerical Monte Carlo simulations, however some later works made 
considerable progress in analyzing such models 
with a variety of analytical methods, see e.g.\ 
\cite{Barrett2,Khalkhali1,Khalkhali2,Khalkhali3,Khalkhali4,Perez-Sanchez1,Perez-Sanchez2,Perez-Sanchez3}
and the review \cite{KhalkhaliReview}.
\sk

The present paper is about a gauge-theoretic study of the dynamical fuzzy spectral triple
models that we have outlined above. Note that such models carry an obvious 
notion of gauge symmetry, described by certain unitary operators $U$ on the Hilbert space $\HH$,
which act via the adjoint action $D\mapsto U \,D\, U^{\ast}$ on the space of Dirac operators $\DD$.
These gauge transformations can be interpreted as noncommutative analogs of the diffeomorphism gauge 
symmetries in ordinary gravity.
Taking into account such gauge symmetries in the definition of the path-integrals 
\eqref{eqn:pathintegral} is an important and non-trivial task which, to the best of our knowledge, 
has not been studied in the literature yet. The standard approach to define
gauge-theoretic path-integrals is through homological methods, such as the 
Batalin–Vilkovisky (BV) formalism \cite{BV} and its 
modern incarnation developed by Costello
and Gwilliam \cite{CostelloGwilliam,CostelloGwilliam2}. In the context
of noncommutative geometry, the BV quantization of field theories on fuzzy spaces has been studied in
\cite{NSSfuzzy} and the classical BV formalism  for a certain type of matrix model 
arising from spectral triples in \cite{IvS17, Ise19a, Ise19b}.
\sk

The main achievement of the present paper is an explicit and computationally accessible
description of both the classical and the perturbative quantum BV formalism for the 
dynamical fuzzy spectral triple models introduced in \cite{Barrett3}. At the classical level, 
we shall construct in particular the extended BV action (involving ghosts and antifields) 
for such models by employing systematic techniques from derived algebraic 
geometry \cite{BSSderived}. At the quantum level, 
we provide a rigorous construction of the perturbative quantum correlation functions 
(including possible ghost and antifield contributions) for quantum perturbations around any background 
Dirac operator $D_0\in\DD$ that solves the classical equations of motion associated with the action $S$.
We then analyze such quantum correlation functions in more detail in order to understand
if the ordinary path-integrals in \eqref{eqn:pathintegral} receive gauge-theoretic modifications
through the ghosts and antifields. We observe that this depends strongly on
the choice of background solution $D_0\in\DD$ one uses to perturb around, or more precisely
on the amount of gauge symmetry that it breaks. In the special case one perturbs around the zero Dirac operator
$D_0=0$, which is gauge invariant, we prove in Proposition \ref{prop:D0perturbation} 
that the ghosts and antifields for ghosts decouple from the correlation functions, 
hence there are no gauge-theoretic modifications to the ordinary path-integrals in \eqref{eqn:pathintegral}.
In stark contrast,  for the generic case of perturbations around a non-trivial Dirac operator $D_0\neq 0$ 
that breaks some of the gauge symmetries, there are non-trivial gauge-theoretic modifications
to the quantum correlation functions. In Section \ref{sec:examples}, 
we shall illustrate this concretely by studying
a simple example, the so-called quartic $(0,1)$-model from \cite{Barrett3}. We observe that this 
model exhibits a Higgs-like mechanism due to its ``symmetry-breaking potential'',
which is identified as the origin of the non-trivial ghost and antifield contributions.
Hence, our results indicate that taking gauge symmetries properly into account
alters the path-integrals in \eqref{eqn:pathintegral}, in particular
in semi-classical situations where the quantum fluctuations are localized 
around a non-trivial classical solution $D_0\neq 0$. An interesting problem for future 
research would be to understand the physical effects associated with these
gauge-theoretic modifications and their interpretation in the context of 
quantum gravity.
\sk

We would like to state very clearly that our present paper
studies gauge-theoretic aspects of fuzzy spectral triple models 
for a {\em fixed} matrix size $N$. Given the highly interesting behavior 
of such models in the large $N$ limit, see e.g.\ the recent review \cite{KhalkhaliReview} 
for an excellent overview, it is natural to ask whether gauge symmetries and BV quantization could have an impact
on large $N$ phenomena such as phase transitions. While this certainly should be expected,
given the gauge-theoretic modifications to finite $N$ correlation functions that
we find in this paper, making mathematically precise statements about the large 
$N$ limit in the BV formalism has been explored very little. 
In a recent series of papers \cite{LargeN1,LargeN2}, the authors
provide a very interesting homological perspective on the
large $N$ limit of the Gaussian Unitary Ensemble from random matrix theory
by establishing a generalization of the Loday-Quillen-Tsygan theorem
that links ordinary BV quantization (as used in the present paper) 
to a noncommutative/cyclic analog. An interesting problem for future research
would be to adapt these techniques to fuzzy spectral triple models, with the hope to 
obtain a homological perspective on the results by the Western University group \cite{KhalkhaliReview}.
In this way it might be possible to determine the effect of gauge symmetries in 
the large $N$ limit on phenomena such as phase transitions.
\sk

The outline of the remainder of this paper is as follows: In Section \ref{sec:spectral},
we provide a brief review of the framework of fuzzy spectral triples from \cite{Barrett1}, 
their gauge symmetries and also their perturbative treatment. The classical BV formalism
for such models is studied in Section \ref{sec:BV}, culminating in an explicit description
of the relevant antibracket and BV action, see Remark \ref{rem:BVaction}. In Section \ref{sec:quantization},
we describe the BV quantization of these models and provide a homological approach
to compute their quantum correlation functions \eqref{eqn:correlation}.
The special case of perturbations around the trivial Dirac operator $D_0=0$
is investigated in detail in Section \ref{sec:D0=0}. In particular, we 
prove that in this case both the ghosts and the antifields for ghosts decouple from the 
quantum correlation functions for observables for the Dirac operator, 
see Proposition \ref{prop:D0perturbation}. In Section \ref{sec:examples},
we shall explain and explicitly show that such decoupling is {\em not} a generic feature
of dynamical fuzzy spectral triples by studying quantum perturbations around a non-trivial Dirac operator $D_0\neq 0$
in the quartic $(0,1)$-model from \cite{Barrett3}. In particular,
we identify and compute to leading order the ghost and antifield contributions to the 
$1$-point and the $2$-point correlation functions of this model, see
Examples \ref{ex:1pt} and \ref{ex:Dnontriv2pt}.


\section{\label{sec:spectral}Fuzzy spectral triples}
The aim of this section is to briefly recall the concept of
{\em fuzzy spectral triples} from \cite{Barrett1},
which are finite-dimensional variants of Connes' real spectral triples \cite{Connes}.
A real spectral triple $(\AA,\HH,\pi,D, \Gamma, J)$ consists of 
a $\ast$-algebra $\AA$ with a $\ast$-representation $\pi:\AA\to\End(\HH)$ on a Hilbert space $\HH$,
a self-adjoint operator $D :\HH\to \HH$ (called Dirac operator), 
a self-adjoint operator $ \Gamma : \HH \rightarrow \HH $ (called chirality operator)
and an anti-unitary operator $ J : \HH \rightarrow \HH $ (called real structure),
which have to satisfy various axioms and compatibility conditions, see e.g.\ 
\cite{Connes,vanSuijlekom,Barrett1} for details.
Fuzzy spectral triples are a particularly simple
class of finite-dimensional real spectral triples, 
for which these data are very explicit:
For the $\ast$-algebra, we fix a natural number $N\in\bbZ^{>0}$ and
take the $N\times N$-matrices with complex entries
\begin{flalign}
\AA \, := \, \Mat_{N}^{}(\bbC)\quad,
\end{flalign}
where the $\ast$-involution is given by Hermitian conjugation.
For the Hilbert space, we choose two non-negative integers
$p,q\in\bbZ^{\geq 0}$, pick a $(p,q)$-Clifford module $V = \bbC^k$
and define
\begin{subequations}
\begin{flalign}
\HH \,:=\, \AA\otimes V
\end{flalign}
together with the Hermitian inner product
\begin{flalign}
\cyc{a\otimes v}{a^\prime \otimes v^\prime}   \,:=\, \Tr_{\AA}^{}(a^\ast a^\prime)~\pair{v}{v^\prime}\quad,
\end{flalign}
\end{subequations}
where $\Tr_\AA^{}$ denotes the trace on the matrix algebra and $\pair{\,\cdot\,}{ \,\cdot\,}$ 
is the standard inner product on $V=\bbC^k$. The $\ast$-representation 
$\pi: \AA\to\End(\HH)$ of $\AA$ on $\HH$ is given by left matrix multiplication
\begin{flalign}
\pi(a^\prime)\big(a\otimes v\big)\,:=\, (a^\prime a)\otimes v\quad.
\end{flalign}
Using the $\bbC$-linear chirality $\gamma: V\to V$ 
and the $\bbC$-anti-linear real structure $C : V \to V$
of the Clifford module $V$, we further define a chirality operator $\Gamma:\HH\to \HH$
and a real structure $J:\HH\to\HH$ on the Hilbert space $\HH$ by setting
\begin{flalign}
\Gamma(a\otimes v) \, :=\, a\otimes \gamma(v)~~,\quad 
J(a\otimes v)\,:=\, a^\ast\otimes C(v)\quad.
\end{flalign}
These two operators satisfy the properties listed in \cite[Definition 1]{Barrett1},
which depend on the KO-dimension $s = q-p\, \, \mathrm{mod}\,\, 8$.
In particular, it follows that $\HH$ is an $\AA$-bimodule with right action
given by
\begin{flalign}
J\,\pi(a^{\prime})^\ast \,J^{-1}\big(a\otimes v\big) \,=\, (a \, a^\prime)\otimes v\quad.
\end{flalign}
We call the tuple $(\AA,\HH, \pi,\Gamma,J)$ defined above
the {\em $(p,q)$-fermion space} over the $N\times N$-matrix algebra $\AA=\Mat_{N}^{}(\bbC)$.
\begin{defi}\label{def:Diracoperator}
A {\em Dirac operator} on the $(p,q)$-fermion space over $\AA=\Mat_{N}^{}(\bbC)$
is an operator $D\in\End(\HH)$ that satisfies the following properties (with $s = q-p\, \, \mathrm{mod}\,\, 8$
the KO-dimension):
\begin{itemize}
\item[(i)] $D^\ast = D$;
\item[(ii)] $D\,\Gamma = - (-1)^{s}\,\Gamma \,D$;
\item[(iii)] $D\,J = \epsilon^\prime\,J\,D$, where $\epsilon^\prime = 1$ for 
$ s=0,2,3,4,6,7$ and $\epsilon^\prime = -1$ for $s=1,5$;
\item[(iv)] $\big[[D,\pi(a)],J \,\pi(b)\,J^{-1}\big] =0$, for all $a,b\in \AA$.
\end{itemize}
We denote the real vector space of Dirac operators by
\begin{flalign}
\DD \,:=\, \big\{D\in \End(\HH)\,:\, \text{(i)-(iv) are satisfied}\big\}\,\subseteq\,\End(\HH)\quad.
\end{flalign}
\end{defi}

As the next step we discuss the group of automorphisms of the $(p,q)$-fermion space
and its action on the space of Dirac operators. An automorphism
of $(\AA,\HH, \pi,\Gamma,J)$ is a pair $(\varphi,\Phi)$
consisting of a $\ast$-algebra automorphism $\varphi: \AA\to\AA$ and 
a left $\AA$-module automorphism $\Phi : \HH\to\HH$ relative to $\varphi$,
i.e.\ $\Phi(a\,h) = \varphi(a)\,\Phi(h)$ for all $a\in\AA$ and $h\in\HH$.
The latter has to preserve the inner product $\cyc{\Phi(\,\cdot\,)}{\Phi(\,\cdot\,)} = \cyc{\,\cdot\,}{\,\cdot\,}$,
the chirality operator $\Gamma\,\Phi =\Phi\,\Gamma$ and the real structure 
$J\,\Phi = \Phi\,J$. Using that $\HH= \AA\otimes V$ is a free module
and that the center of $\AA =\Mat_{N}^{}(\bbC)$ consists of complex multiples of the unit $\oone$,
one easily checks that $\Phi$ must be of the form
\begin{flalign}
\Phi(a\otimes v) \,=\, \varphi(a)\otimes T(v)\quad,
\end{flalign}
for all $a\otimes v\in \HH$, where $T\in \mathrm{Aut}(V)$ 
is an automorphism of the Clifford module $V$ that has to preserve the inner product
$\pair{T(\,\cdot\,)}{T(\,\cdot\,)} = \pair{\,\cdot\,}{\,\cdot\,}$,
the chirality  $\gamma\,T = T \,\gamma$ and the real structure 
$C\,T = T\,C$. Denoting by $K\subseteq \mathrm{Aut}(V)$ the group of all such automorphisms $T$ 
of $V$, we find that the automorphism group of the $(p,q)$-fermion space is isomorphic 
to the product group $\Aut(\AA)\times K$. Note that the two factors
play different roles: The group $\Aut(\AA)$ acts on the underlying $\ast$-algebra and thus
plays the same role as the diffeomorphism group in commutative differential geometry,
whilst the group $K\subseteq \mathrm{Aut}(V)$ acts only on the Clifford module $V$ and thus may
be interpreted as global, i.e.\ $\AA$-independent, transformations of spinors. 
\begin{defi}\label{def:gaugegroup}
We call $\GG := \Aut(\AA)\times K$ the {\em gauge group} of the $(p,q)$-fermion space over $\AA = \Mat_N^{}(\bbC)$.
This group acts from the left as automorphisms of the $(p,q)$-fermion space 
\begin{subequations}
\begin{flalign}
\rho_{\AA}^{}\,:\,\GG\times \AA ~&\longrightarrow~\AA~,~~(\varphi,T,a)~\longmapsto~\rho_{\AA}^{}(\varphi,T)(a)\,=\,\varphi(a)\quad,\\
\rho_{\HH}^{}\,:\,\GG\times \HH ~&\longrightarrow~\HH~,~~(\varphi,T, a\otimes v)~\longmapsto~\rho_{\HH}^{}(\varphi,T)(a\otimes v)\,=\,\varphi(a)\otimes T(v)\quad.
\end{flalign}
\end{subequations}
The induced left adjoint action on the space of Dirac operators is given by
\begin{flalign}
\rho_{\DD}^{}\,:\,\GG\times \DD ~\longrightarrow \DD~,~~(\varphi,T, D)~\longmapsto~  \rho_{\HH}^{}(\varphi,T)\circ D\circ\rho_{\HH}^{}(\varphi^{-1},T^{-1})\quad,
\end{flalign}
where $\circ$ denotes composition of maps.
\end{defi}
\begin{rem}\label{rem:gaugeLiealgebra}
It is well known that $\Aut(\AA) \cong PU(N) := U(N)/U(1)$ is isomorphic to the projective
unitary group, see e.g.\ \cite[Example 6.3]{vanSuijlekom}. The relevant isomorphism
assigns to an element $[u]\in PU(N)$ the automorphism $\varphi_{[u]}\in \Aut(\AA)$
that is defined by $\varphi_{[u]}(a) = u\,a\,u^\ast$, for all $a\in\AA$. 
One may use this isomorphism in order to write the $\GG$-actions from
Definition \ref{def:gaugegroup} more explicitly as $\rho_{\AA}^{}(\varphi_{[u]},T)(a) = u\,a\,u^\ast$
and $\rho_{\HH}^{}(\varphi_{[u]},T)(a\otimes v) = (u\,a\,u^\ast)\otimes T(v)$.
\sk

From this isomorphic perspective,  it is easy to describe the infinitesimal gauge transformations.
The Lie algebra $\g$ of the gauge group $\GG$ is given by a direct sum
\begin{flalign}
\g \,=\, \mathfrak{su}(N)\oplus \mathfrak{k}\quad,
\end{flalign} 
where $\mathfrak{su}(N)\cong \mathfrak{pu}(N)$ is the Lie algebra of the projective unitary group
and $\mathfrak{k}$ denotes the Lie algebra of $K\subseteq \mathrm{Aut}(V)$.
(Recall that elements $\epsilon \in \mathfrak{su}(N)\subseteq \Mat_N^{}(\bbC)$ are anti-Hermitian and trace-free $N\times N$-matrices.)
The Lie algebra actions induced by the $\GG$-actions from Definition \ref{def:gaugegroup}
then read explicitly as
\begin{subequations}\label{eqn:Lieactions}
\begin{flalign}
\rho_{\AA}^{}(\epsilon\oplus k)(a)\,&=\,[\epsilon,a]_{\AA}^{}\quad ,\\
\rho_{\HH}^{}(\epsilon\oplus k)(a\otimes v)\,&=\,[\epsilon ,a]_{\AA}^{}\otimes v  + a\otimes k(v)\quad,\\
\rho_{\DD}^{}(\epsilon \oplus k)(D)\,&=\,[\rho_{\HH}(\epsilon\oplus k),D]_{\End(\HH)}^{}\quad,
\end{flalign}
\end{subequations}
for all $\epsilon \oplus k\in \g$,
where $[\,\cdot\,,\,\cdot\,]_{\AA}^{}$ denotes the commutator on $\AA$
and $[\,\cdot\,,\,\cdot\,]_{\End(\HH)}^{}$  the commutator on $\End(\HH)$.
\end{rem}

The last ingredient is a choice of action $S : \DD\to\bbR$ on the space of Dirac operators,
which is typically taken to be a {\em spectral action} in the sense of 
\cite{SpectralAction1,SpectralAction2}.
The original choice in \cite{Barrett3}, which was motivated by simplicity,
is given by $S(D) = \Tr_{\End(\HH)}(\tfrac{g_2}{2}\,D^2 + \tfrac{g_4}{4!}\,D^4)$,
where $g_2,g_4\in\bbR$ are constants and  $\Tr_{\End(\HH)}$ is the trace on the 
endomorphisms of the Hilbert space $\HH$. This choice was also used in 
\cite{Barrett2,Khalkhali3}. More general choices
of the form $S(D) = \Tr_{\End(\HH)}(f(D))$, where $f$ is a real-valued polynomial, 
were studied later in \cite{Perez-Sanchez1,Khalkhali2,Khalkhali4}
and even more general multi-trace actions appeared in \cite{Khalkhali1}.
For most parts of our paper, we do not have to make any explicit choice for 
the action and can work with the following general definition.
We denote by $\DD^\vee$ the dual of the real vector space $\DD$ 
of Dirac operators and by $\Sym\,\DD^\vee$ its symmetric algebra. Note that
$\Sym\,\DD^\vee$ is (isomorphic to) the algebra of polynomial functions on $\DD$.
\begin{defi}
An {\em action} $S:\DD\to\bbR$ is a gauge-invariant and real-valued polynomial function 
on the space of Dirac operators, or equivalently a $\GG$-invariant element $S\in \Sym\, \DD^{\vee}$.
\end{defi}

To conclude this section, we shall briefly discuss the perturbative treatment
of dynamical fuzzy spectral triples. Let us fix any action $S\in  \Sym\, \DD^{\vee}$
and an exact solution $D_0\in \DD$ of its Euler-Lagrange equations.
We can then study formal perturbations
\begin{flalign}\label{eqn:perturbativeDirac}
D= D_0 + \lambda\, \tilde{D}
\end{flalign} 
of the Dirac operator, where $\lambda$ is a formal parameter and the
perturbation $\tilde{D}\in\DD$ is an element of the same vector space $\DD$. 
The latter is considered as the dynamical field in the perturbative approach.
The infinitesimal gauge transformations \eqref{eqn:Lieactions} of the perturbation then take the form
\begin{flalign}\label{eqn:perturbativegauge}
\tilde{\rho}_{\DD}^{}(\epsilon\oplus k)(\tilde{D}) \,=\, [\rho_\HH(\epsilon\oplus k),D_0]_{\End(\HH)}^{} + \lambda\,
[\rho_\HH(\epsilon\oplus k),\tilde{D}]_{\End(\HH)}^{}\quad,
\end{flalign}
i.e.\ they act through a combination of
a linear transformation $[\rho_\HH(\epsilon\oplus k),\tilde{D}]_{\End(\HH)}^{}$ and an inhomogeneous 
one $ [\rho_\HH(\epsilon\oplus k),D_0]_{\End(\HH)}^{}$ that depends on the background solution $D_0$.
(Note that this is of the same form as in Yang-Mills theory,
where infinitesimal gauge transformations $\chi\in C^\infty(M,\g)$ act on connection one-forms
$A\in \Omega^1(M,\g)$ according to $\delta_\chi(A) = \dd \chi +[\chi,A]$, with $\dd$ the de Rham differential.)
The induced action for the perturbation $\tilde{D}$ is then defined as
\begin{flalign}\label{eqn:perturbativeaction}
\tilde{S}(\tilde{D}) \,:=\,\frac{1}{\lambda^2} \big(S(D_0 + \lambda \,\tilde{D}) - S(D_0)\big)\quad,
\end{flalign}
where the subtraction of the constant term $S(D_0)$ is convenient
as it implies that $\tilde{S}$ is a sum of monomials of degree $\geq 2$.
(The degree $1$ monomial vanishes because $D_0$ is a solution of the Euler-Lagrange equations
for the original action $S$.) The normalization $\frac{1}{\lambda^2}$ is chosen such that the
quadratic term in the action is of order $\lambda^0$.
By construction, this defines an element $\tilde{S}\in\Sym\,\DD^{\vee}$ that is invariant
under the infinitesimal gauge transformations \eqref{eqn:perturbativegauge}.


\section{\label{sec:BV}Classical BV formalism}
The BV formalism is a powerful construction that allows one to
assign a commutative {\em differential graded algebra} (dg-algebra)
of classical observables to every gauge-invariant action. 
Crucially, this dg-algebra comes endowed with a canonical shifted Poisson 
structure, the so-called {\em antibracket}, whose perturbative quantization
captures the quantum correlation functions of the theory.
A modern mathematical perspective on the BV formalism is presented in
the books by Costello and Gwilliam \cite{CostelloGwilliam,CostelloGwilliam2}.
We also refer the reader to \cite{NSSfuzzy} for a specialization
of these techniques to finite-dimensional systems, including fuzzy field theories.
\sk

The goal of this section is to spell out in detail the dg-algebra of classical observables
for the perturbative dynamical fuzzy spectral triple model from Section \ref{sec:spectral}.
This goal can be easily achieved by specializing the general construction in \cite{BSSderived} 
to the following input data:
\begin{itemize}
\item The space of fields is the vector space $\DD$ of perturbations $\tilde{D}$ 
of a background Dirac operator $D_0$ (see \eqref{eqn:perturbativeDirac}) that is an exact solution to
the Euler-Lagrange equations of an action $S$.

\item  The infinitesimal gauge symmetries are given by the Lie algebra
$\g = \mathfrak{su}(N)\oplus \mathfrak{k}$ of the gauge group from Definition \ref{def:gaugegroup}. They 
act on the fields according to \eqref{eqn:perturbativegauge}.

\item The perturbative dynamics is determined by the $\g$-invariant 
action $\tilde{S}$ defined in \eqref{eqn:perturbativeaction}.
\end{itemize}

Applying the general result in \cite[Section 7]{BSSderived}, one finds that the 
underlying $\bbZ$-graded commutative algebra of the BV formalism for this model
is given by the graded symmetric algebra
\begin{flalign}\label{eqn:Obscl}
\Obs^{\cl}\,:=\,\Sym(\LL)~~,\quad \LL\,:=\, \g[2]\oplus \DD[1]\oplus \DD^{\vee} \oplus \g^{\vee}[-1]\quad,
\end{flalign}
where $W^{\vee}$ denotes the dual of a vector space $W$ 
and the square brackets indicate shifts of the cohomological degree. 
(Our convention is that, for $W$ a vector space, elements in $W[p]$ are of degree $-p$.)
We can describe this graded commutative algebra more explicitly 
by choosing a dual pair of vector space bases
\begin{flalign}\label{eqn:basis}
\big\{e_a\in \DD\big\}_{a=1}^{\mathrm{dim}\,\DD}~~,\quad
\big\{f^a\in \DD^\vee\big\}_{a=1}^{\mathrm{dim}\,\DD}~~,\quad
\big\{t_i \in \g \big\}_{i=1}^{\mathrm{dim}\,\g}~~,\quad
\big\{\theta^i \in \g^\vee \big\}_{i=1}^{\mathrm{dim}\,\g}~~
\end{flalign}
for $\DD$ and $\DD^\vee$ and for $\g$ and $\g^\vee$. Then $\Obs^{\cl}$
is the graded commutative algebra generated by the generators $t_i$
in degree $-2$, $e_a$ in degree $-1$, $f^a$ in degree $0$ and $\theta^i$
in degree $1$. The physical interpretation of these generators
is as follows: $\theta^i$ are linear observables for the ghost field $\mathsf{c}\in\g$,
$f^a$ are linear observables for the field $\tilde{D}\in \DD$, 
$e_a$ are linear observables for the antifield $\tilde{D}^+\in\DD^\vee$,
and $t_i$ are linear observables for the antifield for the ghost $\mathsf{c}^+\in\g^\vee$.
\sk

In order to describe the differential on \eqref{eqn:Obscl}, 
which is given abstractly as the totalization of an internal differential
and the Chevalley-Eilenberg differential (see \cite[Section 7]{BSSderived}),
it is convenient to work with these generators. 
Let us first recall that, with respect to our choice of bases, 
the Lie bracket on $\g$ and the Lie algebra 
action \eqref{eqn:perturbativegauge} are encoded by 
structure constants that we denote by
\begin{flalign}\label{eqn:structureconstants}
[t_i,t_j]\,=\, \lambda\, f^{k}_{ij}\,t_k~~,\quad
\tilde{\rho}_{\DD}^{}(t_i)(e_a) = \beta_i^b\,e_b + \lambda\, g_{ia}^b\, e_b\quad,
\end{flalign}
where $\beta_i^b$ describes the inhomogeneous term (depending on $D_0$)
in \eqref{eqn:perturbativegauge} and $g_{ia}^b$ the linear term.
Here, and throughout the whole paper, we use the standard summation convention, i.e.\
we suppress summations over repeated indices. The action $\tilde{S}$ given in \eqref{eqn:perturbativeaction}
also admits a basis expansion of the form
\begin{flalign}\label{eqn:actionbasis}
\tilde{S} = \sum_{n\geq 2} \frac{\lambda^{n-2}}{n!}\, S_{a_1\cdots a_n}\, f^{a_1}\cdots f^{a_n}\,\in \,\Sym\,\DD^{\vee}\quad,
\end{flalign}
which starts at $n=2$ because the background Dirac operator $D_0$ is assumed to be an exact solution
of the Euler-Lagrange equations of the given action $S$.
The coefficients $S_{a_1\cdots a_n}$ vanish for $n$ greater than the polynomial degree of $\tilde{S}$.
The differential on \eqref{eqn:Obscl} is then determined by the graded Leibniz rule
and the following action on the generators
\begin{subequations}\label{eqn:totaldifferential}
\begin{flalign}
\dd t_i \,&=\, \beta^a_i\, e_a + \lambda\, g_{ib}^a\, e_a\,f^b - \lambda\,f_{ij}^k\, t_k\,\theta^j\quad,\\
\dd e_a\,&=\, \sum_{n\geq 2} \frac{\lambda^{n-2}}{(n-1)!} S_{a a_2\cdots a_n} f^{a_2}\cdots f^{a_n} - \lambda\, g_{ja}^b\, e_b\,\theta^j\quad,\\
\dd f^a\,&=\,- \beta^a_j\, \theta^j -\lambda\, g_{jb}^a\, f^b\,\theta^j\quad,\\
\dd \theta^i\,&=\,-\frac{\lambda}{2} f^i_{jk}\,\theta^j\,\theta^k\quad.
\end{flalign}
\end{subequations}
Observe that the differential encodes
both the equation of motion and the gauge symmetries, as typical for the BV formalism. 
The square-zero condition $\dd^2=0$ follows from the Lie algebra representation 
identities for the structure constants $f^k_{ij}$, $g_{ia}^b$ and $\beta^a_i$
and from the fact that the action $\tilde{S}$ is gauge-invariant.
\sk

The dg-algebra $\Obs^{\cl}$ comes endowed with a canonical $(-1)$-shifted symplectic
structure, which in our example reads concretely as
\begin{flalign}
\omega\,=\, \dd^{\dR}e_a\wedge \dd^{\dR} f^a - \dd^{\dR} t_i \wedge \dd^{\dR} \theta^i\quad,
\end{flalign}
where $\dd^{\dR}$ denotes the de Rham differential. The antibracket is
the shifted Poisson bracket dual to $\omega$, which is concretely defined by
$\{a,b\}:=\iota_{{}_aH}\iota_{{}_bH}\omega$, for all $a,b\in \Obs^{\cl}$,
where $\iota$ denotes the contraction between vector fields and forms,
and ${}_aH$ is the shifted Hamiltonian vector field defined by $\dd^{\dR} a = \iota_{{}_aH}\omega$.
By construction, the antibracket satisfies the 
graded antisymmetry property
\begin{subequations}\label{eqn:antibracket}
\begin{flalign}
\{a,b\} \,= \,-(-1)^{(\vert a\vert +1 )\,(\vert b\vert +1)}\,\{b,a\}\quad,
\end{flalign}
the graded Jacobi identity
\begin{flalign}
(-1)^{(\vert a\vert +1)\,(\vert c\vert +1)}\,\{a,\{b,c\}\} + 
(-1)^{(\vert b\vert +1)\,(\vert a\vert +1)}\, \{b,\{c,a\}\} + 
(-1)^{(\vert c\vert +1)\,(\vert b\vert +1)}\, \{c,\{a,b\}\} \,=\, 0\quad,
\end{flalign}
the derivation property
\begin{flalign}\label{eqn:antibracketderivation}
\{a, b\, c\}\,=\, \{a,b\} \,c+ (-1)^{(\vert a\vert +1)\,\vert b\vert} \, b\,\{ a,c\}\quad,
\end{flalign}
as well as the compatibility condition 
\begin{flalign}\label{eqn:antibracketcomp}
\dd \{a,b\} \,=\, \{\dd a,b\} +(-1)^{\vert a\vert +1} \{a,\dd b\} 
\end{flalign}
\end{subequations}
with the differential on $\Obs^{\cl}$.
As a consequence of these properties, the antibracket is completely determined
by its value on the generators, which is given by
\begin{flalign}\label{eqn:P0}
\{t_i, \theta^j\}\,=\, \delta_i^j\,=\, - \{\theta^j,t_i\}~~,\quad
\{e_a, f^b\}\,=\, -\delta_a^b\,=\, -\{f^b,e_a\}\quad,
\end{flalign}
and zero otherwise. 
\begin{rem}\label{rem:BVaction}
The differential \eqref{eqn:totaldifferential} on $\Obs^{\cl}$ admits an
equivalent description in terms of a concept that may be more familiar
to some readers, namely that of a {\em BV action}.
The BV extension of our action in \eqref{eqn:actionbasis}
reads explicitly as
\begin{flalign}
S_{\BV} \,=\, \sum_{n \geq 2} \frac{\lambda^{n-2}}{n!}\, S_{a_{1}\cdots a_{n}}\, f^{a_{1}}\cdots f^{a_{n}}
						- \lambda\, g_{i\, b}^{a}\, e_{a}\, f^{b}\, \theta^{i}
						- \frac{\lambda}{2}\, f_{ij}^{k}\, t_{k}\, \theta^{i}\, \theta^{j}
						- \beta^{a}_{i}\, e_{a}\, \theta^{i}\quad.
\end{flalign}
Using the properties \eqref{eqn:antibracket} of the antibracket, it 
is easy to show that the differential \eqref{eqn:totaldifferential}
is given by 
\begin{flalign}
\dd\,=\, \{ S_{\BV},\,\cdot\, \}\quad.
\end{flalign}
The square-zero condition $\dd^2=0$ for the differential is equivalent to the 
classical master equation
\begin{flalign}
\{S_{\BV},S_{\BV}\}\,=\,0
\end{flalign}
for the BV action.
\end{rem}


\section{\label{sec:quantization}Quantization and correlation functions}
BV quantization consists of deforming the differential \eqref{eqn:totaldifferential}
on $\Obs^{\cl}$ along the BV Laplacian, see e.g.\ \cite{CostelloGwilliam,CostelloGwilliam2,NSSfuzzy}.
The latter is a linear map $\Delta_{\BV} : \Obs^{\cl}\to \Obs^{\cl}$
of cohomological degree $+1$ that is defined on the symmetric powers
$0$, $1$ and $2$ by
\begin{subequations}\label{eqn:BVLaplacian}
\begin{flalign}
\Delta_{\BV} (\oone) \,=\, 0~~,\quad
\Delta_{\BV}(\varphi)\,=\, 0~~,\quad
\Delta_{\BV}(\varphi\,\psi)\,=\, (-1)^{\vert \varphi\vert}\,\{\varphi,\psi\}\quad,
\end{flalign} 
for all generators $\varphi,\psi\in \LL$, and extended to all of $\Obs^\cl$
according to the rule
\begin{flalign}
\Delta_{\BV}(a\,b)\,=\, \Delta_{\BV}(a)\,b + (-1)^{\vert a\vert} \,a\,\Delta_{\BV}(b) + (-1)^{\vert a\vert}\, \{a,b\}\quad,
\end{flalign}
for all $a,b\in\Obs^{\cl}$. Using the algebraic properties \eqref{eqn:antibracket} of the antibracket,
one derives the closed expression
\begin{flalign}
\Delta_{\BV}(\varphi_1\cdots\varphi_n) \,=\,
\sum_{i<j}\, (-1)^{\sum_{k=1}^{i}\,\vert \varphi_k\vert + \vert\varphi_j\vert\,  \sum_{k=i+1}^{j-1}\, \vert \varphi_k\vert}~\{ \varphi_i,\varphi_j \} ~\varphi_1\cdots \check{\varphi}_i\cdots\check{\varphi}_j\cdots\varphi_n 
\end{flalign}
\end{subequations}
for the action of the BV Laplacian on the $n$-th symmetric power of $\Obs^\cl$,
where $\check{\,\cdot\,}$ means to omit the corresponding factor. From this explicit form
one easily checks that the BV Laplacian is square-zero and that
it anti-commutes with the differential $\dd$ on $\Obs^\cl$, i.e.\
\begin{flalign}
\Delta_\BV^2 \,=\, 0~~,\quad \dd\,\Delta_\BV + \Delta_{\BV}\,\dd \,=\,0\quad.
\end{flalign}
The cochain complex of quantum observables is then defined as the deformation
\begin{flalign}\label{eqn:Obsqu}
\Obs^{\qu}\,:=\,\big(\Sym(\LL),\dd^{\qu}:= \dd + \hbar \,\Delta_{\BV}\big)
\end{flalign}
of the one of classical observables given in \eqref{eqn:Obscl} and  \eqref{eqn:totaldifferential},
where $\hbar$ is a formal parameter playing the role of Planck's constant.
\sk

The quantum correlation functions are determined by the cohomology 
of the cochain complex \eqref{eqn:Obsqu} and they can be computed
perturbatively via homological perturbation theory, see e.g.\ 
\cite{Gwilliam} and also \cite{NSSfuzzy} for a review. For this we split
the quantum differential 
\begin{flalign}\label{eqn:differentialqu}
\dd^{\qu}\,=\, \dd^{\fr} + \lambda\,\dd^{\intt} +\hbar \,\Delta_{\BV}
\end{flalign}
into the free part $\dd^{\fr}$, which is obtained by setting $\lambda=0$
in \eqref{eqn:totaldifferential}, the interaction part $\dd^{\intt}$
and the quantum part $\Delta_{\BV}$. Since the free differential $\dd^{\fr}$
is by definition linear in the generators, the classical free observables
\begin{subequations}\label{eqn:Lfree}
\begin{flalign}
\Obs^{\fr}\,:=\,\big(\Sym(\LL),\dd^{\fr} \big) \,=\, \Sym\big(\LL,\dd^\fr\big)
\end{flalign}
are given by the symmetric algebra of the cochain complex 
\begin{flalign}
(\LL,\dd^{\fr})\,=\,\Big(
\xymatrix@C=2.4em{
\g[2] \ar[r]^-{\dd^\fr}& \DD[1]\ar[r]^-{\dd^\fr} & \DD^{\vee} \ar[r]^-{\dd^\fr} & \g^{\vee}[-1]
}\Big)\quad.
\end{flalign}
Explicitly, the action of the differential on our choice of vector space bases
reads as
\begin{flalign}
\dd^{\fr}t_i \,=\,\beta^a_i\, e_a ~~,\quad
\dd^{\fr}e_a\,=\, S_{ab}\,f^b~~,\quad
\dd^{\fr} f^a\,=\, - \beta^a_j\, \theta^j ~~,\quad
\dd^{\fr}\theta^i\,=\,0\quad.
\end{flalign}
\end{subequations}
To apply the homological perturbation lemma,
we choose any strong deformation retract
\begin{equation}\label{eqn:lineardefret}
\begin{tikzcd}
\big( \mathrm{H}^\bullet(\LL,\dd^{\fr}),0 \big) \ar[r,shift right=1ex,swap,"\iota"] & \ar[l,shift right=1ex,swap,"\pi"] 
(\LL,\dd^{\fr})\ar[loop,out=20,in=-20,distance=20,"h"]
\end{tikzcd}
\end{equation}
of the cochain complex $(\LL,\dd^{\fr})$ onto its cohomology.
Recall that the latter consists of two cochain maps, denoted by $\iota$ and $\pi$,
and a cochain homotopy $h$ satisfying the following properties
\begin{flalign}\label{eqn:defretaxioms}
\pi\,\iota\, =\, \id~~,\quad
\iota\,\pi - \id \,=\, \dd^{\fr}\,h + h\,\dd^{\fr}~~,\quad
h\,\iota\, = \, \pi\,h \,=\,  h^2\,=\,0\quad.
\end{flalign}
It was shown in \cite[Proposition 2.5.5]{Gwilliam} that strong deformation
retracts extend to the symmetric algebras
\begin{equation}\label{eqn:classicaldefret}
\begin{tikzcd}
\big( \Sym\, \mathrm{H}^\bullet(\LL,\dd^{\fr}),0 \big) \ar[r,shift right=1ex,swap,"I"] & \ar[l,shift right=1ex,swap,"\Pi"] 
\Obs^\fr_{~} \ar[loop,out=20,in=-20,distance=20,"H"]
\end{tikzcd}\quad.
\end{equation}
Explicitly, the cochain maps $I$ and $\Pi$ are given by extending $\iota$ and $\pi$
to dg-algebra maps via the usual formulas
\begin{flalign}\label{eqn:bigPiI}
I\big([\psi_1]\cdots[\psi_n]\big)\,:=\,\iota([\psi_1])\cdots\iota([\psi_n])~~,\quad
\Pi\big(\varphi_1\cdots\varphi_n\big) \,:=\, \pi(\varphi_1)\cdots\pi(\varphi_n)\quad,
\end{flalign}
for all $[\psi_1],\dots,[\psi_n]\in \mathrm{H}^\bullet(\LL,\dd^{\fr})$ and $\varphi_1,\dots,\varphi_n\in \LL$.
The extended cochain homotopy $H$ is slightly more complicated to describe:
From the definition of strong deformation retract, it follows that $\iota\,\pi: \LL\to \LL$
defines a projector, i.e.\ $(\iota\,\pi)^2 = \iota\,\pi$, which allows us to decompose
\begin{subequations}\label{eqn:Ldirectsum}
\begin{flalign}
\LL\,\cong\, \LL^\perp\oplus \mathrm{H}^\bullet(\LL,\dd^{\fr})
\end{flalign}
and consequently
\begin{flalign}
\Sym(\LL)\,\cong\, \Sym(\LL^\perp)\otimes \Sym\, \mathrm{H}^\bullet(\LL,\dd^{\fr})\,
=\,\bigoplus_{n\geq 0} \Sym^n(\LL^\perp)\otimes \Sym\, \mathrm{H}^\bullet(\LL,\dd^{\fr})\quad,
\end{flalign}
\end{subequations}
where $\Sym^n$ denotes the $n$-th symmetric power. The cochain homotopy $H$
is then defined as
\begin{flalign}\label{eqn:bigH}
H\big(\varphi_1^\perp\,\cdots\,\varphi_n^\perp\otimes a\big)\,:=\,\frac{1}{n}
\sum_{i=0}^n (-1)^{\sum_{j=1}^{i-1}\vert\varphi_j\vert}~
\varphi_1^\perp\,\cdots\, \varphi_{i-1}^{\perp}~h(\varphi_i^\perp)~\varphi_{i+1}^\perp \,\cdots\,\varphi_n^\perp\otimes a
\end{flalign}
on the homogeneous elements $\varphi_1^\perp\,\cdots\,\varphi_n^\perp\otimes a\in 
\Sym^n(\LL^\perp)\otimes \Sym\, \mathrm{H}^\bullet(\LL,\dd^{\fr})$ in this decomposition. 
The case $n=0$ should be read as $H(a)=0$, for all $a\in  \Sym\, \mathrm{H}^\bullet(\LL,\dd^{\fr})$.
The homological perturbation lemma (see e.g.\ \cite{Cra04}) then 
provides a strong deformation retract for the quantum observables.
\begin{propo}
Denote by 
\begin{flalign}
\delta \,:=\, \lambda\,\dd^{\intt} + \hbar \,\Delta_{\BV}
\end{flalign}
the formal perturbation of the free differential $\dd^{\fr}$ into the quantum differential 
$\dd^{\qu}$, see \eqref{eqn:differentialqu}. There exists a deformation 
of the strong deformation retract \eqref{eqn:classicaldefret} into the strong 
deformation retract
\begin{equation}\label{eqn:quantumdefret}
\begin{tikzcd}
\big( \Sym\, \mathrm{H}^\bullet(\LL,\dd^{\fr}),\widetilde{\delta} \big) \ar[r,shift right=1ex,swap,"\widetilde{I}"] & \ar[l,shift right=1ex,swap,"\widetilde{\Pi}"] 
\Obs^\qu_{~} \ar[loop,out=20,in=-20,distance=20,"\widetilde{H}"]
\end{tikzcd}
\end{equation}
for the quantum observables \eqref{eqn:Obsqu}, where
\begin{subequations}
\begin{flalign}
\widetilde{\delta}\,&=\, \Pi\, (\id - \delta  H)^{-1}\, \delta\, I\quad,\\[4pt]
\widetilde{I}\,&=\, I + H\, (\id - \delta  H)^{-1}\,\delta\, I\quad,\\[4pt]
\widetilde{\Pi}\,&=\, \Pi + \Pi\, (\id - \delta  H)^{-1}\,\delta\, H\quad,\\[4pt]
\widetilde{H} \,&=\, H + H\, (\id - \delta  H)^{-1}\, \delta\, H\quad.
\end{flalign}
\end{subequations}
The expression $(\id - \delta H)^{-1} := \sum_{k=0}^\infty (\delta H)^k$ is defined
as a formal power series in both $\lambda$ and $\hbar$.
\end{propo}

The quantum correlation functions are defined by
\begin{flalign}\label{eqn:correlation}
\langle \varphi_1 \cdots\varphi_n\rangle \,:=\,
\widetilde{\Pi}(\varphi_1\cdots\varphi_n)
\,=\, \sum_{k=0}^\infty \Pi\big((\delta H)^k\big(\varphi_1\cdots\varphi_n\big)\big)\,\in\, 
\Sym\, \mathrm{H}^\bullet(\LL,\dd^{\fr})\quad,
\end{flalign}
for all $\varphi_1,\dots,\varphi_n\in \LL$. Note that these are elements
of the symmetric algebra of the cohomology $\mathrm{H}^\bullet(\LL,\dd^{\fr})$,
i.e.\ they are polynomial functions on the space of vacua of the theory, see e.g.\
\cite[Section 3.1]{NSSfuzzy} for an illustration. As we shall illustrate
in the next section, one can use graphical tools to facilitate the computation
of the correlation functions \eqref{eqn:correlation}.


\section{\label{sec:D0=0}Perturbations around $D_0=0$}
In this section we prove some general results about the quantum correlation
functions for perturbations around the zero Dirac operator $D_0=0$.
We fix an arbitrary $(p,q)$-fermion space over $\AA=\Mat_{N}(\bbC)$
and consider any gauge-invariant polynomial action of the form
\begin{flalign}\label{eqn:actionquadratic}
S(D)\,=\, \frac{g_2}{2}\,\Tr_{\End(\HH)}\big(D^2\big) + S^{\intt}(D)\quad,
\end{flalign}
where $g_2\neq 0$ is a non-zero constant and the interaction term $S^\intt$ is a sum of monomials of degree $\geq 3$.
The zero Dirac operator $D_0=0$ is an exact solution of the Euler-Lagrange equation
of this action, hence it is an admissible choice for the background Dirac operator
$D_0$ in our perturbative approach. Our motivation for using the standard choice for the
quadratic term in the action \eqref{eqn:actionquadratic} is as follows:
The complex vector space $\End(\HH)$ can be endowed with the
Hermitian inner product $\langle B,B^\prime\rangle:=
\Tr_{\End(\HH)}(B^\ast\,B^\prime)$, for all $B,B^\prime\in \End(\HH)$.
This restricts to a real inner product on the subspace $\DD\subseteq \End(\HH)$
of Dirac operators defined in Definition \ref{def:Diracoperator}, which implies
that the quadratic term of the action \eqref{eqn:actionquadratic} is non-degenerate.
Hence, as we will see below, it will lead to a particularly simple propagator.
\sk

Let us focus now on the free theory underlying this model \eqref{eqn:Lfree}.
Recalling the definition of the structure constants \eqref{eqn:structureconstants},
we observe that for $D_0=0$ the constants $\beta^a_i=0$ vanish, hence the cochain complex
of linear observables simplifies to
\begin{flalign}
(\LL,\dd^{\fr})\,=\,\Big(
\xymatrix@C=2.4em{
\g[2] \ar[r]^-{0}& \DD[1]\ar[r]^-{\dd^\fr} & \DD^{\vee} \ar[r]^-{0} & \g^{\vee}[-1]
}\Big)\quad,
\end{flalign}
where we recall that $\dd^\fr e_a = S_{ab}\,f^b$ is controlled by the quadratic term of the action.
Because $S_{ab}$ is invertible, one easily checks that the cohomology of this complex is
\begin{flalign}
\mathrm{H}^\bullet(\LL,\dd^{\fr}) \,=\, \g[2]\oplus \g^\vee[-1]\quad. 
\end{flalign}
For the strong deformation retract \eqref{eqn:lineardefret}, we can choose
\begin{flalign}\label{eqn:defretractD0}
\pi\,:\, \begin{cases}
t_i\mapsto t_i\\
e_a\mapsto 0\\
f^a\mapsto 0\\
\theta^i\mapsto \theta^i
\end{cases}\quad,\qquad
\iota\,:\, \begin{cases}
t_i\mapsto t_i\\
\theta^i\mapsto \theta^i
\end{cases}\quad,\qquad
h\,:\,\begin{cases}
t_i\mapsto 0\\
e_a\mapsto 0\\
f^a\mapsto -S^{ab}\,e_b\\
\theta^i\mapsto 0
\end{cases}\quad,
\end{flalign}
where $S^{ab}$ denotes the inverse of $S_{ab}$, i.e.\ 
$S^{ab}S_{bc} = \delta^a_c = S_{cb}S^{ba}$. (The relevant properties 
\eqref{eqn:defretaxioms} are easy to confirm.)
\sk

These are all the necessary ingredients to compute the quantum 
correlation functions for our model.
Indeed, all maps entering \eqref{eqn:correlation}
have been completely defined:
\begin{itemize}
\item[(i)] The cochain homotopy $H$ is defined
by \eqref{eqn:bigH} and its action on generators \eqref{eqn:defretractD0}.
The relevant direct sum decomposition \eqref{eqn:Ldirectsum} in the present case
is given by
\begin{flalign}\label{eqn:Ldirectsum2}
\LL\,=\, \LL^\perp \oplus \mathrm{H}^\bullet(\LL,\dd^{\fr}) \,=\,
\Big(
\xymatrix@C=2.4em{
\DD[1]\ar[r]^-{\dd^\fr} & \DD^{\vee}
}\Big)\oplus \Big( \g[2]\oplus \g^\vee[-1]\Big)\quad.
\end{flalign}

\item[(ii)] The perturbation of the differential is given by $\delta = \lambda\,\dd^\intt + \hbar \,\Delta_{\BV}$.
The interaction part $\dd^\intt$ is defined by the graded Leibniz 
rule and the order $\lambda^{\geq 1}$ terms in \eqref{eqn:totaldifferential}.
The BV Laplacian $\Delta_{\BV}$ is defined by \eqref{eqn:BVLaplacian} and \eqref{eqn:P0}.

\item[(iii)] The dg-algebra map $\Pi$ is defined by \eqref{eqn:bigPiI} and its action on generators 
\eqref{eqn:defretractD0}.
\end{itemize}

To facilitate the computation of correlation functions \eqref{eqn:correlation}, 
we introduce a convenient graphical notation. We denote elements
$\varphi_1\cdots\varphi_n\in \Sym (\LL)$ by $n$ vertical lines. Because there
are four different species of fields (fields, ghosts, antifields and antifields for ghosts),
which are distinguished by their cohomological degree in $\LL$, we require four different
types of lines
\begin{flalign}\label{eqn:linetypes}
t_i\,=~\vcenter{\hbox{\begin{tikzpicture}[scale=0.5]
\draw[thick,dotted] (0,0) -- (0,1.5);
\end{tikzpicture}}}\quad,\qquad
e_a\,=~\vcenter{\hbox{\begin{tikzpicture}[scale=0.5]
\draw[thick,snake it] (0,0) -- (0,1.5);
\end{tikzpicture}}}\quad,\qquad
f^a\,=~\vcenter{\hbox{\begin{tikzpicture}[scale=0.5]
\draw[thick] (0,0) -- (0,1.5);
\end{tikzpicture}}}\quad,\qquad
\theta^i\,=~\vcenter{\hbox{\begin{tikzpicture}[scale=0.5]
\draw[thick,densely dashed] (0,0) -- (0,1.5);
\end{tikzpicture}}}\quad .
\end{flalign}
The action of the cochain homotopy $H$ on $n$ lines can be expressed
via \eqref{eqn:bigH} as a sum of actions on the individual lines. (We would like to emphasize that,
because of \eqref{eqn:Ldirectsum2}, the number $n$ in \eqref{eqn:bigH}
counts only the number of wiggly and straight lines, i.e.\ dashed and dotted lines
do not contribute to $n$.)
Using \eqref{eqn:defretractD0}, we observe that the cochain homotopy $H$ is only non-zero
when acting on $f^a$, which we depict as
\begin{flalign}\label{eqn:homotopypicture}
H\Big(~~\vcenter{\hbox{\begin{tikzpicture}[scale=0.5]
\draw[thick] (0,0) -- (0,1.5);
\end{tikzpicture}}}~~\Big)\,=~
\vcenter{\hbox{\begin{tikzpicture}[scale=0.5]
\draw[thick] (0,0) -- (0,0.75);
\draw[thick,snake it] (0,0.75) -- (0,1.5);
\draw[black,fill=black] (0,0.75) circle (.75ex);
\end{tikzpicture}}}\quad.
\end{flalign}
The action of the interaction part $\dd^{\intt}$ of the differential
on $n$ lines can be expressed via the graded Leibniz rule as a sum of 
actions on the individual lines. The latter may be visualized in terms of
interaction vertices
\begin{subequations}\label{eqn:interactionverticespicture}
\begin{flalign}
\lambda\,\dd^{\intt}\Big(~~\vcenter{\hbox{\begin{tikzpicture}[scale=0.5]
\draw[thick,dotted] (0,0) -- (0,1.5);
\end{tikzpicture}}}~~\Big)\,&=~ \lambda~
\vcenter{\hbox{\begin{tikzpicture}[scale=0.5]
\draw[thick,dotted] (0,0) -- (0,0.75);
\draw[thick,snake it] (0,0.75) -- (-0.75,1.5);
\draw[thick] (0,0.75) -- (0.75,1.5);
\end{tikzpicture}}}~+~\lambda~
\vcenter{\hbox{\begin{tikzpicture}[scale=0.5]
\draw[thick,dotted] (0,0) -- (0,0.75);
\draw[thick,dotted] (0,0.75) -- (-0.75,1.5);
\draw[thick,densely dashed] (0,0.75) -- (0.75,1.5);
\end{tikzpicture}}}\quad,\\[4pt]
\label{eqn:interactionmanylegs}\lambda\,\dd^{\intt}\Big(~~\vcenter{\hbox{\begin{tikzpicture}[scale=0.5]
\draw[thick,snake it] (0,0) -- (0,1.5);
\end{tikzpicture}}}~~\Big)\,&=~ \sum_{n\geq 3} \frac{\lambda^{n-2}}{(n-1)!} \vcenter{\hbox{\begin{tikzpicture}[scale=0.5]
\draw[thick,snake it] (0,0) -- (0,0.75);
\draw[thick] (0,0.75) -- (-0.75,1.5);
\draw[thick] (0,0.75) -- (0.75,1.5);
\node (ov) at (0,2)  {\footnotesize{$n{-}1\text{ legs}$}};
\node (ol) at (0,-0.5)  {\footnotesize{~~}};
\node (dottt) at (0,1.3) {$\cdots$};
\end{tikzpicture}}}~+~\lambda~
\vcenter{\hbox{\begin{tikzpicture}[scale=0.5]
\draw[thick,snake it] (0,0) -- (0,0.75);
\draw[thick,snake it] (0,0.75) -- (-0.75,1.5);
\draw[thick,densely dashed] (0,0.75) -- (0.75,1.5);
\end{tikzpicture}}}\quad,\\[4pt]
\lambda\,\dd^{\intt}\Big(~~\vcenter{\hbox{\begin{tikzpicture}[scale=0.5]
\draw[thick] (0,0) -- (0,1.5);
\end{tikzpicture}}}~~\Big)\,&=~\lambda~
\vcenter{\hbox{\begin{tikzpicture}[scale=0.5]
\draw[thick] (0,0) -- (0,0.75);
\draw[thick] (0,0.75) -- (-0.75,1.5);
\draw[thick,densely dashed] (0,0.75) -- (0.75,1.5);
\end{tikzpicture}}}\quad,\\[4pt]
\lambda\,\dd^{\intt}\Big(~~\vcenter{\hbox{\begin{tikzpicture}[scale=0.5]
\draw[thick,densely dashed] (0,0) -- (0,1.5);
\end{tikzpicture}}}~~\Big)\,&=~\lambda~
\vcenter{\hbox{\begin{tikzpicture}[scale=0.5]
\draw[thick,densely dashed] (0,0) -- (0,0.75);
\draw[thick, densely dashed] (0,0.75) -- (-0.75,1.5);
\draw[thick,densely dashed] (0,0.75) -- (0.75,1.5);
\end{tikzpicture}}}\quad,
\end{flalign}
\end{subequations}
which should be read from bottom to top and whose
numerical values are given in \eqref{eqn:totaldifferential}.
In these pictures, any two neighboring lines can be permuted (up to a Koszul 
sign determined by their cohomological degrees) because they represent elements
in the graded symmetric algebra $\Sym(\LL)$.
The action of the BV Laplacian on $n$ lines can be reduced
via its algebraic properties \eqref{eqn:BVLaplacian} to a sum of pairings
between two lines. Using \eqref{eqn:P0}, we observe that the latter is only non-zero
when pairing between $t_i$ and $\theta^j$ and when pairing between $e_a$ and $f^b$, 
which we depict as
\begin{flalign}\label{eqn:BVLaplacianpictures}
\hbar\,\Delta_{\BV}\Big(~~\vcenter{\hbox{\begin{tikzpicture}[scale=0.5]
\draw[thick,dotted] (0,0) -- (0,1.5);
\draw[thick,densely dashed] (0.75,0) -- (0.75,1.5);
\end{tikzpicture}}}~~\Big)\,=~ \hbar~~
\vcenter{\hbox{\begin{tikzpicture}[scale=0.5]
\draw[thick,dotted] (0,0) -- (0,1.25);
\draw[thick,densely dashed] (0.75,0) -- (0.75,1.25);
\draw[thick,densely dashed] (0,1.25) to[out=90,in=90] (0.75,1.25);
\end{tikzpicture}}}\quad,\qquad
\hbar\,\Delta_{\BV}\Big(~~\vcenter{\hbox{\begin{tikzpicture}[scale=0.5]
\draw[thick,snake it] (0,0) -- (0,1.5);
\draw[thick] (0.75,0) -- (0.75,1.5);
\end{tikzpicture}}}~~\Big)\,=~ \hbar~~
\vcenter{\hbox{\begin{tikzpicture}[scale=0.5]
\draw[thick,snake it] (0,0) -- (0,1.25);
\draw[thick] (0.75,0) -- (0.75,1.25);
\draw[thick] (0,1.25) to[out=90,in=90] (0.75,1.25);
\end{tikzpicture}}}\quad.
\end{flalign}
Again, the ingoing lines can be permuted 
(up to Koszul signs that are trivial in the present case) because they represent
elements in the graded symmetric algebra $\Sym(\LL)$.
Finally, the action of the dg-algebra map $\Pi$ on $n$ lines reduces via
\eqref{eqn:bigPiI} to a product of actions on the individual lines, evaluating
them to their cohomology classes \eqref{eqn:defretractD0}. This is only non-zero
for $t_i$ and $\theta^i$, which we will depict by
\begin{flalign}
\Pi\Big(~~\vcenter{\hbox{\begin{tikzpicture}[scale=0.5]
\draw[thick,dotted] (0,0) -- (0,1.5);
\end{tikzpicture}}}~~\Big)\,=~
\vcenter{\hbox{\begin{tikzpicture}[scale=0.5]
\draw[thick,dotted] (0,0) -- (0,1.5);
\draw[black,fill=white] (0,1.5) circle (.75ex);
\end{tikzpicture}}}\quad,\qquad
\Pi\Big(~~\vcenter{\hbox{\begin{tikzpicture}[scale=0.5]
\draw[thick,densely dashed] (0,0) -- (0,1.5);
\end{tikzpicture}}}~~\Big)\,=~
\vcenter{\hbox{\begin{tikzpicture}[scale=0.5]
\draw[thick,densely dashed] (0,0) -- (0,1.5);
\draw[black,fill=white] (0,1.5) circle (.75ex);
\end{tikzpicture}}}\quad.
\end{flalign}

\begin{ex}\label{ex:Dis0}
To illustrate the diagrammatic calculus, let us consider 
a quartic interaction term (i.e.\ the $n = 4$ term is the only non-trivial 
summand in \eqref{eqn:interactionmanylegs}) 
and compute the $2$-point correlation function
\begin{flalign}
\langle \varphi_1\, \varphi_2\rangle \,=\,
\sum_{k=0}^\infty \Pi\big((\delta H)^k\big(\varphi_1\, \varphi_2\big)\big)
\quad, \quad\quad \varphi_{1},\varphi_{2} \in \LL^{0} = \DD^\vee\quad,
\end{flalign}
for two generators in degree $0$ to the 
lowest non-trivial order in the formal parameter $\lambda$. 
The first step of this iterative process is computing
\begin{flalign}
\nn	\delta H\big(\varphi_1\, \varphi_2\big) &\,=~ \frac{1}{2}\,\delta
	\Big(~~
	\vcenter{\hbox{\begin{tikzpicture}[scale=0.5]
		\draw[thick] (0,0) -- (0,0.75);
		\draw[thick,snake it] (0,0.75) -- (0,1.5);
		\draw[black,fill=black] (0,0.75) circle (.75ex);
		\draw[thick] (0.75,0) -- (0.75,1.5);
	\end{tikzpicture}}}
	~+~
	\vcenter{\hbox{\begin{tikzpicture}[scale=0.5]
		\draw[thick] (0,0) -- (0,1.5);
		\draw[thick] (0.75,0) -- (0.75,0.75);
		\draw[thick,snake it] (0.75,0.75) -- (0.75,1.5);
		\draw[black,fill=black] (0.75,0.75) circle (.75ex);
	\end{tikzpicture}}}
	~~\Big)\\[4pt]
\nn	&\,=~
	\frac{\hbar}{2}\,\Big(~~
	\vcenter{\hbox{\begin{tikzpicture}[scale=0.5]
		\draw[thick] (0,0) -- (0,0.75);
		\draw[thick,snake it] (0,0.75) -- (0,1.5);
		\draw[black,fill=black] (0,0.75) circle (.75ex);
		\draw[thick] (0.75,0) -- (0.75,1.5);
		\draw[thick] (0,1.5) to[out=90,in=90] (0.75,1.5);
	\end{tikzpicture}}}
	~+~
	\vcenter{\hbox{\begin{tikzpicture}[scale=0.5]
		\draw[thick] (0,0) -- (0,1.5);
		\draw[thick] (0,1.5) to[out=90,in=90] (0.75,1.5);	
		\draw[thick] (0.75,0) -- (0.75,0.75);
		\draw[thick,snake it] (0.75,0.75) -- (0.75,1.5);
		\draw[black,fill=black] (0.75,0.75) circle (.75ex);
	\end{tikzpicture}}}
	~~\Big)
	~~+~\frac{\lambda^{2}}{3!\,2}~~
	\vcenter{\hbox{\begin{tikzpicture}[scale=0.5]
		\draw[thick] (0,0) -- (0,0.75);
		\draw[thick,snake it] (0,0.75) -- (0,1.5);
		\draw[black,fill=black] (0,0.75) circle (.75ex);
		\draw[thick] (0,1.5) -- (-0.75,2.25);
		\draw[thick] (0,1.5) -- (0,2.25);
		\draw[thick] (0,1.5) -- (0.75,2.25);
		\draw[thick] (1.5,0) -- (1.5,1.5);
	\end{tikzpicture}}}
	~~+~\frac{\lambda}{2}~~
	\vcenter{\hbox{\begin{tikzpicture}[scale=0.5]
		\draw[thick] (0,0) -- (0,0.75);
		\draw[thick,snake it] (0,0.75) -- (0,1.5);
		\draw[black,fill=black] (0,0.75) circle (.75ex);
		\draw[thick,snake it] (0,1.5) -- (-0.75,2.25);
		\draw[thick,densely dashed] (0,1.5) -- (0.75,2.25);
		\draw[thick] (1.5,0) -- (1.5,1.5);
	\end{tikzpicture}}}
	~~-~\frac{\lambda}{2}~~
	\vcenter{\hbox{\begin{tikzpicture}[scale=0.5]
		\draw[thick] (0,0) -- (0,0.75);
		\draw[thick,snake it] (0,0.75) -- (0,1.5);
		\draw[black,fill=black] (0,0.75) circle (.75ex);
		\draw[thick] (1.5,0) -- (1.5,1.5);
		\draw[thick] (1.5,1.5) -- (0.75,2.25);
		\draw[thick,densely dashed] (1.5,1.5) -- (2.25,2.25);
	\end{tikzpicture}}}\\[4pt]
\nn	&\quad\quad~~+~ \frac{\lambda}{2}~~
	\vcenter{\hbox{\begin{tikzpicture}[scale=0.5]
		\draw[thick] (0,0) -- (0,1.5);
		\draw[thick] (0,1.5) -- (-0.75,2.25);
		\draw[thick,densely dashed] (0,1.5) -- (0.75,2.25);
		\draw[thick] (1.5,0) -- (1.5,0.75);
		\draw[thick,snake it] (1.5,0.75) -- (1.5,1.5);
		\draw[black,fill=black] (1.5,0.75) circle (.75ex);
	\end{tikzpicture}}}
	~~+~ \frac{\lambda^{2}}{3!\,2}~~
	\vcenter{\hbox{\begin{tikzpicture}[scale=0.5]
		\draw[thick] (0,0) -- (0,1.5);
		\draw[thick] (1.5,0) -- (1.5,0.75);
		\draw[thick,snake it] (1.5,0.75) -- (1.5,1.5);
		\draw[black,fill=black] (1.5,0.75) circle (.75ex);
		\draw[thick] (1.5,1.5) -- (0.75,2.25);
		\draw[thick] (1.5,1.5) -- (1.5,2.25);
		\draw[thick] (1.5,1.5) -- (2.25,2.25);
	\end{tikzpicture}}}
	~~+~ \frac{\lambda}{2}~~
	\vcenter{\hbox{\begin{tikzpicture}[scale=0.5]
		\draw[thick] (0,0) -- (0,1.5);
		\draw[thick] (1.5,0) -- (1.5,0.75);
		\draw[thick,snake it] (1.5,0.75) -- (1.5,1.5);
		\draw[black,fill=black] (1.5,0.75) circle (.75ex);
		\draw[thick,snake it] (1.5,1.5) -- (0.75,2.25);
		\draw[thick,densely dashed] (1.5,1.5) -- (2.25,2.25);
	\end{tikzpicture}}}\\[4pt]
\nn	&\,=~
	\hbar~~
	\vcenter{\hbox{\begin{tikzpicture}[scale=0.5]
		\draw[thick] (0,0) -- (0,0.75);
		\draw[thick,snake it] (0,0.75) -- (0,1.5);
		\draw[black,fill=black] (0,0.75) circle (.75ex);
		\draw[thick] (0.75,0) -- (0.75,1.5);
		\draw[thick] (0,1.5) to[out=90,in=90] (0.75,1.5);
	\end{tikzpicture}}}
	~~+~~ \frac{\lambda^{2}}{3!\,2}\,\Bigg(~~
	\vcenter{\hbox{\begin{tikzpicture}[scale=0.5]
		\draw[thick] (0,0) -- (0,0.75);
		\draw[thick,snake it] (0,0.75) -- (0,1.5);
		\draw[black,fill=black] (0,0.75) circle (.75ex);
		\draw[thick] (0,1.5) -- (-0.75,2.25);
		\draw[thick] (0,1.5) -- (0,2.25);
		\draw[thick] (0,1.5) -- (0.75,2.25);
		\draw[thick] (1.5,0) -- (1.5,1.5);
	\end{tikzpicture}}}
	~~+~~
	\vcenter{\hbox{\begin{tikzpicture}[scale=0.5]
		\draw[thick] (0,0) -- (0,1.5);
		\draw[thick] (1.5,0) -- (1.5,0.75);
		\draw[thick,snake it] (1.5,0.75) -- (1.5,1.5);
		\draw[black,fill=black] (1.5,0.75) circle (.75ex);
		\draw[thick] (1.5,1.5) -- (0.75,2.25);
		\draw[thick] (1.5,1.5) -- (1.5,2.25);
		\draw[thick] (1.5,1.5) -- (2.25,2.25);
	\end{tikzpicture}}}
	~~\Bigg)\\[4pt]
	&\quad\quad ~~+~~ \frac{\lambda}{2}~
	\Bigg(~~
	\vcenter{\hbox{\begin{tikzpicture}[scale=0.5]
		\draw[thick] (0,0) -- (0,0.75);
		\draw[thick,snake it] (0,0.75) -- (0,1.5);
		\draw[black,fill=black] (0,0.75) circle (.75ex);
		\draw[thick,snake it] (0,1.5) -- (-0.75,2.25);
		\draw[thick,densely dashed] (0,1.5) -- (0.75,2.25);
		\draw[thick] (1.5,0) -- (1.5,1.5);
	\end{tikzpicture}}}
	~~-~~
	\vcenter{\hbox{\begin{tikzpicture}[scale=0.5]
		\draw[thick] (0,0) -- (0,0.75);
		\draw[thick,snake it] (0,0.75) -- (0,1.5);
		\draw[black,fill=black] (0,0.75) circle (.75ex);
		\draw[thick] (1.5,0) -- (1.5,1.5);
		\draw[thick] (1.5,1.5) -- (0.75,2.25);
		\draw[thick,densely dashed] (1.5,1.5) -- (2.25,2.25);
	\end{tikzpicture}}}
	~~+~~
	\vcenter{\hbox{\begin{tikzpicture}[scale=0.5]
		\draw[thick] (0,0) -- (0,1.5);
		\draw[thick] (0,1.5) -- (-0.75,2.25);
		\draw[thick,densely dashed] (0,1.5) -- (0.75,2.25);
		\draw[thick] (1.5,0) -- (1.5,0.75);
		\draw[thick,snake it] (1.5,0.75) -- (1.5,1.5);
		\draw[black,fill=black] (1.5,0.75) circle (.75ex);
	\end{tikzpicture}}}
	~~+~~
	\vcenter{\hbox{\begin{tikzpicture}[scale=0.5]
		\draw[thick] (0,0) -- (0,1.5);
		\draw[thick] (1.5,0) -- (1.5,0.75);
		\draw[thick,snake it] (1.5,0.75) -- (1.5,1.5);
		\draw[black,fill=black] (1.5,0.75) circle (.75ex);
		\draw[thick,snake it] (1.5,1.5) -- (0.75,2.25);
		\draw[thick,densely dashed] (1.5,1.5) -- (2.25,2.25);
	\end{tikzpicture}}}
	~~\Bigg)\quad.
\end{flalign}
The negative sign in the second equality is a Koszul sign arising from the fact that 
$\delta = \lambda\,\dd^\intt + \hbar \,\Delta_{\BV}$ is of cohomological degree $1$.
In the third equality, we have collected the interaction terms according to their power in $\lambda$.
The simplification of the $\hbar$-term is a consequence of the fact that
\begin{subequations}\label{eqn:homotopypull}
\begin{flalign}
	\vcenter{\hbox{\begin{tikzpicture}[scale=0.5]
		\draw[thick] (0,0) -- (0,0.75);
		\draw[thick,snake it] (0,0.75) -- (0,1.5);
		\draw[black,fill=black] (0,0.75) circle (.75ex);
		\draw[thick] (0.75,0) -- (0.75,1.5);
		\draw[thick] (0,1.5) to[out=90,in=90] (0.75,1.5);
	\end{tikzpicture}}}
	~=\, \Delta_{\BV}\big(h(\varphi_1)\, \varphi_{2}\big)
	\,=\, -\{h(\varphi_{1}),\varphi_{2}\} \,=\, \{\varphi_{1},h(\varphi_{2})\}
	\,=\, \Delta_{\BV}\big(\varphi_1\, h(\varphi_{2})\big) \,=~
	\vcenter{\hbox{\begin{tikzpicture}[scale=0.5]
		\draw[thick] (0,0) -- (0,1.5);
		\draw[thick] (0,1.5) to[out=90,in=90] (0.75,1.5);	
		\draw[thick] (0.75,0) -- (0.75,0.75);
		\draw[thick,snake it] (0.75,0.75) -- (0.75,1.5);
		\draw[black,fill=black] (0.75,0.75) circle (.75ex);
	\end{tikzpicture}}}\quad,
\end{flalign}
where the middle step is easily checked on basis elements
\begin{flalign}
-\{h(f^a),f^b\}\,=\, S^{ac}\,\{e_c,f^b\}\,= \, -S^{ab} \,=\, -S^{ba}\,=\,
-S^{bc}\{f^a,e_c\}\,= \,\{f^a,h(f^b)\}
\end{flalign}
\end{subequations}
by using \eqref{eqn:P0}, \eqref{eqn:defretractD0} and symmetry of $S^{ab}$.
\sk

Using similar arguments and the fact that the interaction term 
$S_{a_{1} a_2 a_3 a_{4}}$ is symmetric under the exchange of indices, 
the second iteration is given by
\begin{flalign}
\nn	(\delta H)^2\big(\varphi_1\, \varphi_2\big) &\,=~
	\frac{\hbar\,\lambda^{2}}{8}\Bigg(~~
	\vcenter{\hbox{\begin{tikzpicture}[scale=0.5]
		\draw[thick] (0,0) -- (0,0.75);
		\draw[thick,snake it] (0,0.75) -- (0,1.5);
		\draw[black,fill=black] (0,0.75) circle (.75ex);
		\draw[thick] (0,1.5) -- (-0.75,2.25);
		\draw[black,fill=black] (-0.75,2.25) circle (.75ex);
		\draw[thick,snake it] (-0.75,2.25) -- (-0.75,3);
		\draw[thick] (-0.75,3) to[out=90,in=90] (0,3);
		\draw[thick] (0,1.5) -- (0,3);
		\draw[thick] (0,1.5) -- (0.75,2.25);
		\draw[thick] (1.5,0) -- (1.5,2.25);
	\end{tikzpicture}}}
	~~+~~
	\vcenter{\hbox{\begin{tikzpicture}[scale=0.5]
		\draw[thick] (0,0) -- (0,0.75);
		\draw[thick,snake it] (0,0.75) -- (0,1.5);
		\draw[black,fill=black] (0,0.75) circle (.75ex);
		\draw[thick] (0,1.5) -- (-0.75,2.25);
		\draw[thick] (0,1.5) -- (0,2.25);
		\draw[thick] (0,1.5) -- (0.75,2.25);
		\draw[thick] (1.5,0) -- (1.5,1.125);
		\draw[thick,snake it] (1.5,1.125) -- (1.5,2.25);
		\draw[black,fill=black] (1.5,1.125) circle (.75ex);
		\draw[thick] (0.75,2.25) to[out=90,in=90] (1.5,2.25);
	\end{tikzpicture}}}
	~~+~~
	\vcenter{\hbox{\begin{tikzpicture}[scale=0.5]
		\draw[thick] (0,0) -- (0,1.125);
		\draw[thick,snake it] (0,1.125) -- (0,2.25);
		\draw[black,fill=black] (0,1.125) circle (.75ex);
		\draw[thick] (0,2.25) to[out=90,in=90] (0.75,2.25);
		\draw[thick] (1.5,0) -- (1.5,0.75);
		\draw[thick,snake it] (1.5,0.75) -- (1.5,1.5);
		\draw[black,fill=black] (1.5,0.75) circle (.75ex);
		\draw[thick] (1.5,1.5) -- (0.75,2.25);
		\draw[thick] (1.5,1.5) -- (1.5,2.25);
		\draw[thick] (1.5,1.5) -- (2.25,2.25);
	\end{tikzpicture}}}
	~~+~~
	\vcenter{\hbox{\begin{tikzpicture}[scale=0.5]
		\draw[thick] (0,0) -- (0,2.25);
		\draw[thick] (1.5,0) -- (1.5,0.75);
		\draw[thick,snake it] (1.5,0.75) -- (1.5,1.5);
		\draw[black,fill=black] (1.5,0.75) circle (.75ex);
		\draw[thick] (1.5,1.5) -- (0.75,2.25);
		\draw[thick] (1.5,1.5) -- (1.5,3);
		\draw[thick] (1.5,1.5) -- (2.25,2.25);
		\draw[thick,snake it] (2.25,2.25) -- (2.25,3);
		\draw[black,fill=black] (2.25,2.25) circle (.75ex);
		\draw[thick] (1.5,3) to[out=90,in=90] (2.25,3);
	\end{tikzpicture}}}
	~~\Bigg)\\[4pt]
\nn	&\quad\quad ~+~\frac{\lambda^2}{4}~\Bigg(~~
	\vcenter{\hbox{\begin{tikzpicture}[scale=0.5]
		\draw[thick] (0,0) -- (0,0.75);
		\draw[thick,snake it] (0,0.75) -- (0,1.5);
		\draw[black,fill=black] (0,0.75) circle (.75ex);
		\draw[thick,snake it] (0,1.5) -- (-0.75,2.25);
		\draw[thick,densely dashed] (0,1.5) -- (0.75,2.25);
		\draw[thick,snake it] (-0.75,2.25) -- (-1.5,3);
		\draw[thick,densely dashed] (-0.75,2.25) -- (0,3);
		\draw[thick] (1.5,0) -- (1.5,1.125);
		\draw[thick,snake it] (1.5,1.125) -- (1.5,2.25);
		\draw[black,fill=black] (1.5,1.125) circle (.75ex);
	\end{tikzpicture}}}
	~~-~~
	\vcenter{\hbox{\begin{tikzpicture}[scale=0.5]
		\draw[thick] (0,0) -- (0,0.75);
		\draw[thick,snake it] (0,0.75) -- (0,1.5);
		\draw[black,fill=black] (0,0.75) circle (.75ex);
		\draw[thick,snake it] (0,1.5) -- (-0.75,2.25);
		\draw[thick,densely dashed] (0,1.5) -- (0.75,2.25);
		\draw[thick,densely dashed] (0.75,2.25) -- (0,3);
		\draw[thick,densely dashed] (0.75,2.25) -- (1.5,3);
		\draw[thick] (2.25,0) -- (2.25,1.125);
		\draw[thick,snake it] (2.25,1.125) -- (2.25,2.25);
		\draw[black,fill=black] (2.25,1.125) circle (.75ex);
	\end{tikzpicture}}}
	~~+~~
	\vcenter{\hbox{\begin{tikzpicture}[scale=0.5]
		\draw[thick] (0,0) -- (0,0.75);
		\draw[thick,snake it] (0,0.75) -- (0,1.5);
		\draw[black,fill=black] (0,0.75) circle (.75ex);
		\draw[thick,snake it] (0,1.5) -- (-0.75,2.25);
		\draw[thick,densely dashed] (0,1.5) -- (0.75,2.25);
		\draw[thick] (2.25,0) -- (2.25,0.75);
		\draw[thick,snake it] (2.25,0.75) -- (2.25,1.5);
		\draw[black,fill=black] (2.25,0.75) circle (.75ex);
		\draw[thick,snake it] (2.25,1.5) -- (1.5,2.25);
		\draw[thick,densely dashed] (2.25,1.5) -- (3,2.25);
	\end{tikzpicture}}}
	~~+~~
	\vcenter{\hbox{\begin{tikzpicture}[scale=0.5]
		\draw[thick] (0,0) -- (0,0.75);
		\draw[thick,snake it] (0,0.75) -- (0,1.5);
		\draw[black,fill=black] (0,0.75) circle (.75ex);
		\draw[thick,snake it] (0,1.5) -- (-0.75,2.25);
		\draw[thick,densely dashed] (0,1.5) -- (0.75,2.25);
		\draw[thick] (2.25,0) -- (2.25,0.75);
		\draw[thick] (2.25,0.75) -- (1.5,1.5);
		\draw[thick,densely dashed] (2.25,0.75) -- (3,1.5);
		\draw[thick,snake it] (1.5,1.5) -- (1.5,2.25);
		\draw[black,fill=black] (1.5,1.5) circle (.75ex);
	\end{tikzpicture}}}\\[4pt]
\nn	&\quad\quad\quad\quad\quad\quad ~~-~~
	\vcenter{\hbox{\begin{tikzpicture}[scale=0.5]
		\draw[thick] (0,0) -- (0,1.125);
		\draw[thick,snake it] (0,1.125) -- (0,2.25);
		\draw[black,fill=black] (0,1.125) circle (.75ex);
		\draw[thick] (2.25,0) -- (2.25,0.75);
		\draw[thick] (2.25,0.75) -- (1.5,1.5);
		\draw[thick,densely dashed] (2.25,0.75) -- (3,1.5);
		\draw[thick,snake it] (1.5,1.5) -- (1.5,2.25);
		\draw[black,fill=black] (1.5,1.5) circle (.75ex);
		\draw[thick,snake it] (1.5,2.25) -- (0.75,3);
		\draw[thick,densely dashed] (1.5,2.25) -- (2.25,3);
	\end{tikzpicture}}}
	~~+~~
	\vcenter{\hbox{\begin{tikzpicture}[scale=0.5]
		\draw[thick] (0,0) -- (0,1.125);
		\draw[thick,snake it] (0,1.125) -- (0,2.25);
		\draw[black,fill=black] (0,1.125) circle (.75ex);
		\draw[thick] (1.5,0) -- (1.5,0.75);
		\draw[thick] (1.5,0.75) -- (0.75,1.5);
		\draw[thick,densely dashed] (1.5,0.75) -- (2.25,1.5);
		\draw[thick,snake it] (0.75,1.5) -- (0.75,2.25);
		\draw[black,fill=black] (0.75,1.5) circle (.75ex);
		\draw[thick,densely dashed] (2.25,1.5) -- (1.5,2.25);
		\draw[thick,densely dashed] (2.25,1.5) -- (3,2.25);
	\end{tikzpicture}}}
	~~+~~
	\vcenter{\hbox{\begin{tikzpicture}[scale=0.5]
		\draw[thick] (0,0) -- (0,0.75);
		\draw[thick] (0,0.75) -- (-0.75,1.5);
		\draw[thick,densely dashed] (0,0.75) -- (0.75,1.5);
		\draw[thick,snake it] (-0.75,1.5) -- (-0.75,2.25);
		\draw[black,fill=black] (-0.75,1.5) circle (.75ex);
		\draw[thick,snake it] (-0.75,2.25) -- (-1.5,3);
		\draw[thick,densely dashed] (-0.75,2.25) -- (0,3);
		\draw[thick] (1.5,0) -- (1.5,1.125);
		\draw[thick,snake it] (1.5,1.125) -- (1.5,2.25);
		\draw[black,fill=black] (1.5,1.125) circle (.75ex);
	\end{tikzpicture}}}
	~~-~~
	\vcenter{\hbox{\begin{tikzpicture}[scale=0.5]
		\draw[thick] (0,0) -- (0,0.75);
		\draw[thick] (0,0.75) -- (-0.75,1.5);
		\draw[thick,densely dashed] (0,0.75) -- (0.75,1.5);
		\draw[thick,snake it] (-0.75,1.5) -- (-0.75,2.25);
		\draw[black,fill=black] (-0.75,1.5) circle (.75ex);
		\draw[thick,densely dashed] (0.75,1.5) -- (0,2.25);
		\draw[thick,densely dashed] (0.75,1.5) -- (1.5,2.25);
		\draw[thick] (2.25,0) -- (2.25,1.125);
		\draw[thick,snake it] (2.25,1.125) -- (2.25,2.25);
		\draw[black,fill=black] (2.25,1.125) circle (.75ex);
	\end{tikzpicture}}}
	\\[4pt]
\nn	&\quad\quad\quad\quad\quad\quad ~~+~~
	\vcenter{\hbox{\begin{tikzpicture}[scale=0.5]
		\draw[thick] (0,0) -- (0,0.75);
		\draw[thick] (0,0.75) -- (-0.75,1.5);
		\draw[thick,densely dashed] (0,0.75) -- (0.75,1.5);
		\draw[thick,snake it] (-0.75,1.5) -- (-0.75,2.25);
		\draw[black,fill=black] (-0.75,1.5) circle (.75ex);
		\draw[thick] (2.25,0) -- (2.25,0.75);
		\draw[thick,snake it] (2.25,0.75) -- (2.25,1.5);
		\draw[black,fill=black] (2.25,0.75) circle (.75ex);
		\draw[thick,snake it] (2.25,1.5) -- (1.5,2.25);
		\draw[thick,densely dashed] (2.25,1.5) -- (3,2.25);
	\end{tikzpicture}}}
	~~+~~
	\vcenter{\hbox{\begin{tikzpicture}[scale=0.5]
		\draw[thick] (0,0) -- (0,0.75);
		\draw[thick,snake it] (0,0.75) -- (0,1.5);
		\draw[black,fill=black] (0,0.75) circle (.75ex);
		\draw[thick,snake it] (0,1.5) -- (-0.75,2.25);
		\draw[thick,densely dashed] (0,1.5) -- (0.75,2.25);
		\draw[thick] (2.25,0) -- (2.25,0.75);
		\draw[thick,snake it] (2.25,0.75) -- (2.25,1.5);
		\draw[black,fill=black] (2.25,0.75) circle (.75ex);
		\draw[thick,snake it] (2.25,1.5) -- (1.5,2.25);
		\draw[thick,densely dashed] (2.25,1.5) -- (3,2.25);
	\end{tikzpicture}}}
	~~-~~
	\vcenter{\hbox{\begin{tikzpicture}[scale=0.5]
		\draw[thick] (0,0) -- (0,1.125);
		\draw[thick,snake it] (0,1.125) -- (0,2.25);
		\draw[black,fill=black] (0,1.125) circle (.75ex);
		\draw[thick] (2.25,0) -- (2.25,0.75);
		\draw[thick,snake it] (2.25,0.75) -- (2.25,1.5);
		\draw[black,fill=black] (2.25,0.75) circle (.75ex);
		\draw[thick,snake it] (2.25,1.5) -- (1.5,2.25);
		\draw[thick,densely dashed] (2.25,1.5) -- (3,2.25);
		\draw[thick,snake it] (1.5,2.25) -- (0.75,3);
		\draw[thick,densely dashed] (1.5,2.25) -- (2.25,3);
	\end{tikzpicture}}}
	~~+~~
	\vcenter{\hbox{\begin{tikzpicture}[scale=0.5]
		\draw[thick] (0,0) -- (0,1.125);
		\draw[thick,snake it] (0,1.125) -- (0,2.25);
		\draw[black,fill=black] (0,1.125) circle (.75ex);
		\draw[thick] (1.5,0) -- (1.5,0.75);
		\draw[thick,snake it] (1.5,0.75) -- (1.5,1.5);
		\draw[black,fill=black] (1.5,0.75) circle (.75ex);
		\draw[thick,snake it] (1.5,1.5) -- (0.75,2.25);
		\draw[thick,densely dashed] (1.5,1.5) -- (2.25,2.25);
		\draw[thick,densely dashed] (2.25,2.25) -- (1.5,3);
		\draw[thick,densely dashed] (2.25,2.25) -- (3,3);
	\end{tikzpicture}}}
	~~\Bigg)\\[4pt]
	&\quad\quad ~~+~~ \mathcal{O}(\lambda^{3})
\end{flalign}
and the third by
\begin{flalign}
	(\delta H)^3\big(\varphi_1\, \varphi_2\big) \,=~
	\frac{\hbar^{2}\lambda^{2}}{2}~~
	\vcenter{\hbox{\begin{tikzpicture}[scale=0.5]
		\draw[thick] (0,0) -- (0,0.75);
		\draw[thick,snake it] (0,0.75) -- (0,1.5);
		\draw[black,fill=black] (0,0.75) circle (.75ex);
		\draw[thick] (0,1.5) -- (-0.75,2.25);
		\draw[black,fill=black] (-0.75,2.25) circle (.75ex);
		\draw[thick,snake it] (-0.75,2.25) -- (-0.75,3);
		\draw[thick] (-0.75,3) to[out=90,in=90] (0,3);
		\draw[thick] (0,1.5) -- (0,3);
		\draw[thick] (0,1.5) -- (0.75,2.25);
		\draw[thick] (1.5,0) -- (1.5,1.125);
		\draw[thick,snake it] (1.5,1.125) -- (1.5,2.25);
		\draw[black,fill=black] (1.5,1.125) circle (.75ex);
		\draw[thick] (0.75,2.25) to[out=90,in=90] (1.5,2.25);
	\end{tikzpicture}}}
	~~+~~ \mathcal{O}(\lambda^{3})
	\quad.
\end{flalign}
Applying $\Pi$ to these expressions, we 
obtain that the $2$-point function to order $\lambda^{2}$ is given by
\begin{flalign}
	\langle \varphi_1\, \varphi_2\rangle \,=~ \hbar~~
	\vcenter{\hbox{\begin{tikzpicture}[scale=0.5]
		\draw[thick] (0,0) -- (0,0.75);
		\draw[thick,snake it] (0,0.75) -- (0,1.5);
		\draw[black,fill=black] (0,0.75) circle (.75ex);
		\draw[thick] (0.75,0) -- (0.75,1.5);
		\draw[thick] (0,1.5) to[out=90,in=90] (0.75,1.5);
	\end{tikzpicture}}}
	~~+~~
	\frac{\hbar^{2}\lambda^{2}}{2}~~
	\vcenter{\hbox{\begin{tikzpicture}[scale=0.5]
		\draw[thick] (0,0) -- (0,0.75);
		\draw[thick,snake it] (0,0.75) -- (0,1.5);
		\draw[black,fill=black] (0,0.75) circle (.75ex);
		\draw[thick] (0,1.5) -- (-0.75,2.25);
		\draw[black,fill=black] (-0.75,2.25) circle (.75ex);
		\draw[thick,snake it] (-0.75,2.25) -- (-0.75,3);
		\draw[thick] (-0.75,3) to[out=90,in=90] (0,3);
		\draw[thick] (0,1.5) -- (0,3);
		\draw[thick] (0,1.5) -- (0.75,2.25);
		\draw[thick] (1.5,0) -- (1.5,1.125);
		\draw[thick,snake it] (1.5,1.125) -- (1.5,2.25);
		\draw[black,fill=black] (1.5,1.125) circle (.75ex);
		\draw[thick] (0.75,2.25) to[out=90,in=90] (1.5,2.25);
	\end{tikzpicture}}}
	~~+~~ \mathcal{O}(\lambda^{3})
	\quad.
\end{flalign}
Observe that neither the ghosts nor the antifields for ghosts contribute 
to the $2$-point function at order $\lambda^{2}$. 
We will prove below that, for the present
case of perturbations around the zero Dirac operator $D_0=0$,
this statement is true for {\em all} $n$-point functions of degree
zero observables, {\em all} interaction terms and to {\em all} 
orders of the perturbation series.
\end{ex}
\begin{propo}\label{prop:D0perturbation}
Consider an arbitrary $(p,q)$-fermion space over $\AA=\Mat_{N}(\bbC)$
and any gauge-invariant action of the form \eqref{eqn:actionquadratic}.
Then, for perturbations around the zero Dirac operator $D_0=0$, 
all $n$-point quantum correlation functions
$\langle \varphi_1\cdots\varphi_n\rangle$ for degree $0$ observables
$\varphi_1,\dots,\varphi_n\in \LL^0 = \DD^\vee$ do {\em not} 
receive contributions from ghosts and antifields for ghosts.
\end{propo}
\begin{proof}
The proof is a simple argument using our graphical calculus.
Starting from $n$ straight lines, representing the element
$\varphi_1\cdots\varphi_n\in \Sym(\LL)$, one observes 
by direct inspection that iterated applications
of $\delta H$ do {\em not} include any 
dotted (antifield for ghost) lines
due to the explicit form of the
interaction vertices \eqref{eqn:interactionverticespicture}
and of the cochain homotopy \eqref{eqn:homotopypicture}.
Because $\delta H$ is of cohomological degree $0$,
the element $(\delta H)^k(\varphi_1\cdots\varphi_n)\in \Sym(\LL)$
is of cohomological degree $0$ too, hence together with the previous
observation it must contain an equal number of dashed (ghost) lines
and wiggly (antifield) lines. Applying the dg-algebra homomorphism
$\Pi$ and using that it gives zero on every wiggly line, it follows
that only those terms with no ghost lines contribute to the correlation function
$\langle \varphi_1\cdots\varphi_n\rangle = \sum_{k=0}^\infty\Pi\big((\delta H)^k(\varphi_1\cdots\varphi_n)\big)$.
This completes the proof.
\end{proof}


\section{\label{sec:examples}Example for $D_0\neq 0$: The quartic $(0,1)$-model}
The aim of this section is to study quantized perturbations
around non-zero Dirac operators $D_0\neq 0$ in the
simplest dynamical fuzzy spectral triple model,
the so-called $(0,1)$-model from \cite{Barrett1}.
From this we can show that, in contrast to the case of 
perturbations around $D_0=0$ as in Proposition \ref{prop:D0perturbation}, 
the quantum correlation functions for perturbations around $D_0\neq 0$ are in general 
sensitive to the ghosts and the antifields for ghosts.
By working through the computational details of this simple model,
it will become evident that this behavior is not accidental,
but it will be present in any other model provided that
one perturbs around a background Dirac operator $D_0$ that 
breaks some of the gauge symmetries. (Note that the zero Dirac operator
$D_0=0$ is gauge invariant under \eqref{eqn:Lieactions}.)

\paragraph{The $(0,1)$-model:} Let us recall from \cite{Barrett1,Barrett3} 
that the $(0,1)$-fermion space
over $\AA=\Mat_{N}^{}(\bbC)$ is given by the Hilbert space
$\HH=\AA$, on which $\AA$ acts via left multiplication, 
with inner product $\cyc{a}{a^\prime}  = \Tr_{\AA}^{}(a^\ast a^\prime)$,
chirality operator $\Gamma(a) = a$ and real structure $\Gamma(a) = a^\ast$.
Furthermore, each Dirac operator (in the sense of Definition \ref{def:Diracoperator})
is of the form 
\begin{flalign}\label{eqn:DviaL}
D \,=\, -\ii\,[L,\,\cdot\,]\quad,
\end{flalign}
where $L\in \Mat_{N}^{}(\bbC) $ is an anti-Hermitian and trace-free $N\times N$-matrix.
Hence, we can identify 
\begin{flalign}
\DD\,\cong\,\su(N)
\end{flalign}
the space of Dirac operators on the $(0,1)$-fermion space with (the underlying real vector space of) 
the Lie algebra $\su(N)$ of anti-Hermitian and trace-free $N\times N$-matrices.
From Remark \ref{rem:gaugeLiealgebra}, we obtain that the Lie algebra
of infinitesimal gauge transformations of this model is given by
\begin{flalign}
\g\,=\,\su(N)\quad,
\end{flalign}
which acts on the space of Dirac operators through the Lie bracket of $\su(N)$, i.e.
\begin{flalign}\label{eqn:Lieaction01}
\rho_{\DD}^{}\, :\,\g\times \DD ~\longrightarrow~\DD~,~~(\epsilon,L)~\longmapsto~\rho_{\DD}^{}(\epsilon)(L)
\,=\,[\epsilon,L]\quad.
\end{flalign}
We consider as in \cite{Barrett3} the gauge-invariant quartic
action $S(D) = \Tr_{\End(\HH)}^{}\big(\tfrac{g_2}{2} D^2 + \tfrac{g_4}{4!} D^4\big) $
and assume that $g_2<0$ is negative and $g_4>0$ is positive in order to obtain a ``symmetry-breaking potential''.
Upon rescaling the Dirac operator $D$, one can assume without loss of generality that $g_2=-1$.
Inserting \eqref{eqn:DviaL} in this action and using that $L$ is trace-free, one can
express $S$ as a function of $L$ and finds
\begin{flalign}\label{eqn:action01}
S(L)\,=\,N\,\Tr_{\AA}^{}(L^2) + \frac{g_4}{4!}\Big(2\,N\,\Tr_{\AA}^{}(L^4) + 6 \,\big(\Tr_{\AA}^{}(L^2)\big)^2\Big)\quad.
\end{flalign}
The variation of this action reads as
\begin{subequations}
\begin{flalign}\label{eqn:variation01}
\delta S(L)\,=\, \Tr_{\AA}^{}\Big[\delta L ~\Big(2\, N\, L + \frac{g_4}{4!}\big(8\, N\, L^3 + 24\, L\,\Tr_{\AA}^{}(L^2)\big)\Big) \Big]\quad,
\end{flalign}
which yields the Euler-Lagrange equation
\begin{flalign}\label{eqn:ELequations01}
L + \frac{g_4}{3!}\Big(L^3 - \frac{1}{N} \,\Tr_{\AA}^{}(L^3) + \frac{3}{N}\, L\,\Tr_{\AA}^{}(L^2)\Big)\,=\,0\quad.
\end{flalign}
\end{subequations}
Note that the term $ \Tr_{\AA}^{}(L^3)$ in \eqref{eqn:ELequations01} arises due to the fact
that the variation $\delta L$ is trace-free, hence the big round bracket
in \eqref{eqn:variation01} must be projected onto the space of trace-free and anti-Hermitian matrices.
\sk

In order to illustrate the main features arising for perturbations around a non-zero Dirac operator,
we will consider the case where $N$ is even and the following simple exact solution
\begin{flalign}\label{eqn:L0background01}
L_0 \,=\, \ii\,\kappa \,\begin{pmatrix}
\oone_{\sfrac{N}{2}}^{}&0\\
0& - \oone_{\sfrac{N}{2}}^{}
\end{pmatrix}\quad,\qquad \kappa\,:=\, \sqrt{\frac{3}{2\,g_4}} 
\end{flalign}
of \eqref{eqn:ELequations01}, which we have written in block matrix notation.
The novel feature of this background solution is that it breaks the $\g=\su(N)$
gauge symmetry down to a Lie sub-algebra. Indeed, using also
a block matrix notation for Lie algebra elements
\begin{flalign}
\epsilon\,=\,\begin{pmatrix}
\epsilon_1 & \epsilon_3\\
-\epsilon_3^\ast &\epsilon_2 
\end{pmatrix}\,\in\,\g \,=\,\su(N)\quad,
\end{flalign}
one finds for the Lie algebra action \eqref{eqn:Lieaction01} on $L_0$ that
\begin{flalign}\label{eqn:gaction01}
\rho_{\DD}^{}(\epsilon)(L_0)\,=\,[\epsilon,L_0] \,=\, -2\,\ii\,\kappa\,\begin{pmatrix}
0 & \epsilon_3\\
\epsilon_3^\ast & 0
\end{pmatrix}\quad.
\end{flalign}
Hence, the $\su(N)$ gauge symmetry is broken down to the Lie sub-algebra 
$\g_0 \subset \g=\su(N)$ whose elements are of the form
\begin{flalign}\label{eqn:subLie01}
\epsilon = \begin{pmatrix}
\epsilon_1&0\\
0 & \epsilon_2
\end{pmatrix}\,\in\,\g_0 \subset \g = \su(N)\quad.
\end{flalign}
For later use, let us note that the linearization around $L_0$ of the Euler-Lagrange
equation \eqref{eqn:ELequations01} reads as
\begin{subequations}\label{eqn:EOMlin01}
\begin{flalign}
P(\tilde{L})\,:=\,-\frac{1}{2} \tilde{L} + \frac{g_4}{3!}\, L_0\,[\tilde{L},L_0] 
+ \frac{g_4}{N} \,L_0\,\Tr_{\AA}^{}(\tilde{L}\,L_0)\,=\,0\quad.
\end{flalign}
In block matrix notation, 
we can write this more explicitly as
\begin{flalign}
P(\tilde{L})  = -\frac{1}{2}\begin{pmatrix}
\tilde{L}_1 & 0\\
0 & \tilde{L}_2
\end{pmatrix} -\frac{3}{2N}\,\Tr(\tilde{L}_1-\tilde{L}_2)\,\begin{pmatrix}
\oone_{\sfrac{N}{2}} & 0\\
0 & -\oone_{\sfrac{N}{2}}
\end{pmatrix}\quad,
\end{flalign}
\end{subequations}
where the unadorned trace $\Tr$ is over $\sfrac{N}{2}\times\sfrac{N}{2}$-matrices.
Furthermore, the perturbative action \eqref{eqn:perturbativeaction} around $L_0$ 
reads as
\begin{flalign}
\nn \tilde{S}(\tilde{L})\,&=\, -\frac{N}{4}\,\Tr_{\AA}^{}(\tilde{L}^2) + \frac{g_4}{4!}\Big(
4\,N\,\Tr_{\AA}^{}(L_0\,\tilde{L}\,L_0\,\tilde{L}) + 24\, \big(\Tr_{\AA}^{}(L_0\,\tilde{L})\big)^2\Big)\\
\nn &\quad +\frac{\lambda \,g_4}{4!}\Big(8\,N\,\Tr_{\AA}^{}(L_0\,\tilde{L}^3) + 24\, \Tr_{\AA}^{}(L_0\,\tilde{L})\,\Tr_{\AA}^{}(\tilde{L}^2)\Big)\\
&\quad + \frac{\lambda^2\,g_4}{4!}\Big(2\,N\,\Tr_{\AA}^{}(\tilde{L}^4) + 6\,\big(\Tr_{\AA}^{}(\tilde{L}^2)\big)^2\Big)\quad. \label{eqn:perturbativeaction01}
\end{flalign}

\paragraph{The complex of linear observables \eqref{eqn:Lfree}:} Recalling 
that $\DD\cong \su(N)$ and $\g=\su(N)$, we can use the Killing form
$2N\,\Tr_\AA^{}(X\,Y)$, for $X,Y\in\su(N)$, to identify
$\DD^{\vee}\cong \su(N)$ and $\g^\vee \cong \su(N)$.
The cochain complex \eqref{eqn:Lfree} of linear observables then 
takes for our $(0,1)$-model the form
\begin{subequations}\label{eqn:lincomplex01}
\begin{flalign}
\LL \,=\,\Big(\xymatrix@C=2.4em{
\su(N)[2] \ar[r]^-{\dd^\fr}& \su(N)[1]\ar[r]^-{\dd^\fr} & \su(N) \ar[r]^-{\dd^\fr} & \su(N)[-1]
}
\Big)\quad.
\end{flalign}
To avoid confusing the different $\su(N)$-components in this complex, we shall use the following
notations
\begin{flalign}
\beta\,\in\, \su(N)[2]~~,\quad \alpha\,\in\,\su(N)[1]~~,\quad
\varphi\,\in\, \su(N)~~,\quad \chi\,\in\,\su(N)[-1]
\end{flalign}
\end{subequations}
and recall that (in the presented order) they are interpreted as linear
observables for the antifield for the ghost $\mathsf{c}^+$, the antifield $\tilde{D}^+$,
the field $\tilde{D}$ and the ghost $\mathsf{c}$. The free differential
can be expressed very efficiently by re-writing \eqref{eqn:Lfree} 
in terms of our block matrix notation (see \eqref{eqn:gaction01} and \eqref{eqn:EOMlin01}), 
which yields
\begin{subequations}\label{eqn:lincomplex01differential}
\begin{flalign}
\dd^{\fr}\beta \,&=\, [\beta,L_0]\,=\, -2\,\ii\,\kappa\,\begin{pmatrix}
0 & \beta_3\\
\beta_3^\ast & 0
\end{pmatrix}\quad,\\
\dd^{\fr}\alpha \,&=\,P(\alpha)\,=\,-\frac{1}{2}\begin{pmatrix}
\alpha_1 & 0\\
0 & \alpha_2
\end{pmatrix} -\frac{3}{2N}\,\Tr(\alpha_1-\alpha_2)\,\begin{pmatrix}
\oone_{\sfrac{N}{2}} & 0\\
0 & -\oone_{\sfrac{N}{2}}
\end{pmatrix}\quad,\\
\dd^{\fr}\varphi\,&=\, [\varphi,L_0]\,=\, - 2\,\ii\,\kappa\,\begin{pmatrix}
0 & \varphi_3\\
\varphi_3^\ast & 0
\end{pmatrix}\quad,\\
\dd^{\fr}\chi\,&=\, 0\quad.
\end{flalign}
\end{subequations}
Note that there is no relative sign between $\dd^{\fr}\beta$ 
and $\dd^{\fr}\varphi$ because the minus sign in \eqref{eqn:Lfree} 
is compensated by the minus sign that comes
from forming adjoints with respect to the Killing form.
(Explicitly, $2N\,\Tr_\AA^{}(X\,[Y,\epsilon]) = -2N\,\Tr_{\AA}^{}([X,\epsilon]\,Y)$.)
\sk

Recalling the Lie sub-algebra $\g_0\subset \g =\su(N)$ introduced in \eqref{eqn:subLie01},
we can further simplify the description of the complex of linear observables given in
\eqref{eqn:lincomplex01} and \eqref{eqn:lincomplex01differential}.
For this let us also introduce the orthogonal complement $\g_0^{\perp}\subset \g=\su(N)$
of $\g_0\subset \g =\su(N)$ with respect to the Killing form. Note that,
in our block matrix notation, the elements of the orthogonal complement are of the form
\begin{flalign}\label{eqn:subLie01complement}
\epsilon = \begin{pmatrix}
0&\epsilon_3\\
-\epsilon_3^\ast & 0
\end{pmatrix}\,\in\,\g_0^\perp \subset \g = \su(N)\quad.
\end{flalign}
Decomposing all four $\su(N)$-components in \eqref{eqn:lincomplex01} 
according to $\su(N) = \g_0 \oplus \g_0^{\perp}$, we
observe that each component of the differential in \eqref{eqn:lincomplex01differential} 
is only non-zero on one of the two summands. This allows us to decompose the complex
\eqref{eqn:lincomplex01} as
\begin{flalign}\label{eqn:lincomplex01simple}
\LL \,=\,
\left(\parbox{3cm}{\xymatrix@R=0.1em@C=2.4em{
\g_0[2] & \g_0[1] \ar[r]^-{\dd^{\fr}}& \g_0 & \g_0[-1]\\
\oplus&\oplus&\oplus&\oplus\\
\g_0^\perp[2] \ar[r]_-{\dd^{\fr}}& \g_0^\perp[1] & \g_0^\perp \ar[r]_-{\dd^{\fr}}& \g_0^\perp[-1]
}}\right)\quad.
\end{flalign}
The displayed non-trivial differentials are all injective 
and hence, via the rank nullity theorem, also surjective on 
the corresponding summands. This implies that the cohomology of this complex
is given by
\begin{flalign}\label{eqn:cohomology01}
\mathrm{H}^\bullet\big(\LL,\dd^{\fr}\big)\,=\, \g_0[2]\oplus \g_{0}[-1]\quad.
\end{flalign}

\paragraph{The strong deformation retract \eqref{eqn:lineardefret}:}
From \eqref{eqn:lincomplex01simple}, we obtain a direct sum decomposition
\begin{subequations}\label{eqn:LLdecomposition01}
\begin{flalign}
\LL \,=\, \LL^\perp\oplus \mathrm{H}^\bullet\big(\LL,\dd^{\fr}\big)
\end{flalign}
into the acyclic complex
\begin{flalign}
\LL^\perp\,=\,\left(\parbox{3cm}{\xymatrix@R=0.1em@C=2.4em{
 & \g_0[1] \ar[r]^-{\dd^{\fr}}& \g_0 & \\
&\oplus&\oplus&\\
\g_0^\perp[2] \ar[r]_-{\dd^{\fr}}& \g_0^\perp[1] & \g_0^\perp \ar[r]_-{\dd^{\fr}}& \g_0^\perp[-1]
}}\right)
\end{flalign}
\end{subequations}
and the cohomology \eqref{eqn:cohomology01}.
The $\pi$-map of the desired strong deformation retract \eqref{eqn:lineardefret}
is then simply given by projecting onto $\mathrm{H}^\bullet(\LL,\dd^{\fr})$
and the $\iota$-map is given by inclusion. A choice of a contracting homotopy 
$h$ can be found by inverting (the non-trivial parts of) $\dd^{\fr}$, yielding
\begin{subequations}\label{eqn:homotopy01}
\begin{flalign}
h(\beta)\,&=\, 0\quad,\\
h(\alpha)\,&=\, -\frac{\ii}{2\, \kappa}\begin{pmatrix}
0 & \alpha_3\\
\alpha_3^\ast & 0
\end{pmatrix}\quad,\\
h(\varphi) \,&=\, 2\begin{pmatrix}
\varphi_{1} & 0\\
0 & \varphi_{2}
\end{pmatrix} -\frac{3}{2N}\,\Tr(\varphi_1-\varphi_2)\,\begin{pmatrix}
\oone_{\sfrac{N}{2}} & 0\\
0 & -\oone_{\sfrac{N}{2}}
\end{pmatrix}\quad,\\
h(\chi) \,&=\, -\frac{\ii}{2\, \kappa}\begin{pmatrix}
0 & \chi_{3}\\
\chi_3^\ast & 0
\end{pmatrix}\quad.
\end{flalign}
\end{subequations}
The relevant properties in \eqref{eqn:defretaxioms} are straightforward to verify.

\paragraph{Quantum correlation functions:} The diagrammatic approach
from Section \ref{sec:D0=0} to compute the quantum correlation 
functions \eqref{eqn:correlation} generalizes with slight adjustments
to perturbations around non-zero Dirac operators. In analogy to \eqref{eqn:linetypes}, there are 
four types of lines, which we further decompose according
to \eqref{eqn:LLdecomposition01} into their $\LL^{\perp}$ and cohomology components.
Graphically, we will write this as
\begin{flalign}
\beta \,=~\vcenter{\hbox{\begin{tikzpicture}[scale=0.5]
\draw[thick,dotted] (0,0) -- (0,1.5);
\end{tikzpicture}}}~=~\vcenter{\hbox{\begin{tikzpicture}[scale=0.5]
\draw[thick,dotted] (0,0) -- (0,1.5) node[midway, left] {\!\!{\footnotesize ${}_{\perp}$}};
\end{tikzpicture}}}~+~
\vcenter{\hbox{\begin{tikzpicture}[scale=0.5]
\draw[thick,dotted] (0,0) -- (0,1.5) node[midway, left] {\!\!{\footnotesize ${}_{0}$}};
\end{tikzpicture}}}\quad,\qquad
\alpha \,=~\vcenter{\hbox{\begin{tikzpicture}[scale=0.5]
\draw[thick,snake it] (0,0) -- (0,1.5);
\end{tikzpicture}}}\quad,\qquad
\varphi \,=~\vcenter{\hbox{\begin{tikzpicture}[scale=0.5]
\draw[thick] (0,0) -- (0,1.5);
\end{tikzpicture}}}\quad,\qquad
\chi \,=~\vcenter{\hbox{\begin{tikzpicture}[scale=0.5]
\draw[thick,densely dashed] (0,0) -- (0,1.5);
\end{tikzpicture}}}
~=~\vcenter{\hbox{\begin{tikzpicture}[scale=0.5]
\draw[thick,densely dashed] (0,0) -- (0,1.5) node[midway, left] {\!\!{\footnotesize ${}_{\perp}$}};
\end{tikzpicture}}}~+~
\vcenter{\hbox{\begin{tikzpicture}[scale=0.5]
\draw[thick,densely dashed] (0,0) -- (0,1.5) node[midway, left] {\!\!{\footnotesize ${}_{0}$}};
\end{tikzpicture}}}\quad .
\end{flalign}
The action of the cochain homotopy $H$ on many lines can be expressed
via \eqref{eqn:bigH} as a sum of actions on the individual lines. 
Note that, by definition, the number $n$ in this formula does {\em not} count the 
cohomology components of the lines, i.e.\ it only counts
the wiggly, straight, $\perp$-dashed and $\perp$-dotted lines.
The homotopy on a single line is given in \eqref{eqn:homotopy01}
and it is only non-zero on the $\LL^{\perp}$-components.
We graphically represent the non-vanishing components of the homotopy as
\begin{flalign}
H\Big(~~\vcenter{\hbox{\begin{tikzpicture}[scale=0.5]
\draw[thick, densely dashed] (0,0) -- (0,1.5) node[midway, left] {\!\!\!\!\!\!{\footnotesize ${}_{\perp}$}};
\end{tikzpicture}}}~~\Big)\,=~
\vcenter{\hbox{\begin{tikzpicture}[scale=0.5]
\draw[thick, densely dashed] (0,0) -- (0,0.75) node[midway, left] {\!\!{\footnotesize ${}_{\perp}$}};
\draw[thick] (0,0.75) -- (0,1.5);
\draw[black,fill=black] (0,0.75) circle (.75ex);
\end{tikzpicture}}}\quad,\qquad
H\Big(~~\vcenter{\hbox{\begin{tikzpicture}[scale=0.5]
\draw[thick] (0,0) -- (0,1.5);
\end{tikzpicture}}}~~\Big)\,=~
\vcenter{\hbox{\begin{tikzpicture}[scale=0.5]
\draw[thick] (0,0) -- (0,0.75);
\draw[thick,snake it] (0,0.75) -- (0,1.5);
\draw[black,fill=black] (0,0.75) circle (.75ex);
\end{tikzpicture}}}\quad,\qquad
H\Big(~~\vcenter{\hbox{\begin{tikzpicture}[scale=0.5]
\draw[thick, snake it] (0,0) -- (0,1.5);
\end{tikzpicture}}}~~\Big)\,=~
\vcenter{\hbox{\begin{tikzpicture}[scale=0.5]
\draw[thick, snake it] (0,0) -- (0,0.75);
\draw[thick, dotted] (0,0.75) -- (0,1.5) node[midway, left] {\!\!{\footnotesize ${}_{\perp}$}};
\draw[black,fill=black] (0,0.75) circle (.75ex);
\end{tikzpicture}}}\quad.
\end{flalign}
The interaction part $\dd^{\intt}$ of the differential for our model is
given by specializing \eqref{eqn:interactionverticespicture}
to the perturbative action \eqref{eqn:perturbativeaction01}. This yields cubic 
and quartic interaction vertices, which we visualize as
\begin{subequations}
\begin{flalign}
\lambda\,\dd^{\intt}\Big(~~\vcenter{\hbox{\begin{tikzpicture}[scale=0.5]
\draw[thick,dotted] (0,0) -- (0,1.5);
\end{tikzpicture}}}~~\Big)\,&=~ \lambda~
\vcenter{\hbox{\begin{tikzpicture}[scale=0.5]
\draw[thick,dotted] (0,0) -- (0,0.75);
\draw[thick,snake it] (0,0.75) -- (-0.75,1.5);
\draw[thick] (0,0.75) -- (0.75,1.5);
\end{tikzpicture}}}~+~\lambda~
\vcenter{\hbox{\begin{tikzpicture}[scale=0.5]
\draw[thick,dotted] (0,0) -- (0,0.75);
\draw[thick,dotted] (0,0.75) -- (-0.75,1.5);
\draw[thick,densely dashed] (0,0.75) -- (0.75,1.5);
\end{tikzpicture}}}\quad,\\[4pt]
\lambda\,\dd^{\intt}\Big(~~\vcenter{\hbox{\begin{tikzpicture}[scale=0.5]
\draw[thick,snake it] (0,0) -- (0,1.5);
\end{tikzpicture}}}~~\Big)\,&=~ \frac{\lambda}{2}~
\vcenter{\hbox{\begin{tikzpicture}[scale=0.5]
\draw[thick,snake it] (0,0) -- (0,0.75);
\draw[thick] (0,0.75) -- (-0.75,1.5);
\draw[thick] (0,0.75) -- (0.75,1.5);
\end{tikzpicture}}}~+~
\frac{\lambda^2}{3!}~
\vcenter{\hbox{\begin{tikzpicture}[scale=0.5]
\draw[thick,snake it] (0,0) -- (0,0.75);
\draw[thick] (0,0.75) -- (-0.75,1.5);
\draw[thick] (0,0.75) -- (0.75,1.5);
\draw[thick] (0,0.75) -- (0,1.5);
\end{tikzpicture}}}~+~
\lambda~
\vcenter{\hbox{\begin{tikzpicture}[scale=0.5]
\draw[thick,snake it] (0,0) -- (0,0.75);
\draw[thick,snake it] (0,0.75) -- (-0.75,1.5);
\draw[thick,densely dashed] (0,0.75) -- (0.75,1.5);
\end{tikzpicture}}}\quad,\\[4pt]
\lambda\,\dd^{\intt}\Big(~~\vcenter{\hbox{\begin{tikzpicture}[scale=0.5]
\draw[thick] (0,0) -- (0,1.5);
\end{tikzpicture}}}~~\Big)\,&=~\lambda~
\vcenter{\hbox{\begin{tikzpicture}[scale=0.5]
\draw[thick] (0,0) -- (0,0.75);
\draw[thick] (0,0.75) -- (-0.75,1.5);
\draw[thick,densely dashed] (0,0.75) -- (0.75,1.5);
\end{tikzpicture}}}\quad,\\[4pt]
\lambda\,\dd^{\intt}\Big(~~\vcenter{\hbox{\begin{tikzpicture}[scale=0.5]
\draw[thick,densely dashed] (0,0) -- (0,1.5);
\end{tikzpicture}}}~~\Big)\,&=~\lambda~
\vcenter{\hbox{\begin{tikzpicture}[scale=0.5]
\draw[thick,densely dashed] (0,0) -- (0,0.75);
\draw[thick, densely dashed] (0,0.75) -- (-0.75,1.5);
\draw[thick,densely dashed] (0,0.75) -- (0.75,1.5);
\end{tikzpicture}}}\quad.
\end{flalign}
\end{subequations}
The numerical values of these interaction vertices with respect 
to a choice of bases are given in \eqref{eqn:totaldifferential}.
The BV Laplacian is precisely the one from Section \ref{sec:D0=0}, which we again visualize as
\begin{flalign}
\hbar\,\Delta_{\BV}\Big(~~\vcenter{\hbox{\begin{tikzpicture}[scale=0.5]
\draw[thick,dotted] (0,0) -- (0,1.5);
\draw[thick,densely dashed] (0.75,0) -- (0.75,1.5);
\end{tikzpicture}}}~~\Big)\,=~ \hbar~~
\vcenter{\hbox{\begin{tikzpicture}[scale=0.5]
\draw[thick,dotted] (0,0) -- (0,1.25);
\draw[thick,densely dashed] (0.75,0) -- (0.75,1.25);
\draw[thick,densely dashed] (0,1.25) to[out=90,in=90] (0.75,1.25);
\end{tikzpicture}}}\quad,\qquad
\hbar\,\Delta_{\BV}\Big(~~\vcenter{\hbox{\begin{tikzpicture}[scale=0.5]
\draw[thick,snake it] (0,0) -- (0,1.5);
\draw[thick] (0.75,0) -- (0.75,1.5);
\end{tikzpicture}}}~~\Big)\,=~ \hbar~~
\vcenter{\hbox{\begin{tikzpicture}[scale=0.5]
\draw[thick,snake it] (0,0) -- (0,1.25);
\draw[thick] (0.75,0) -- (0.75,1.25);
\draw[thick] (0,1.25) to[out=90,in=90] (0.75,1.25);
\end{tikzpicture}}}\quad,
\end{flalign}
and the $\Pi$-map acts only non-trivially on the cohomology components of the lines, i.e.\
\begin{flalign}
\Pi\Big(~~\vcenter{\hbox{\begin{tikzpicture}[scale=0.5]
\draw[thick, dotted] (0,0) -- (0,1.5) node[midway, left] {\!\!\!\!\!\!{\footnotesize ${}_{0}$}};
\end{tikzpicture}}}~~\Big)\,=~
\vcenter{\hbox{\begin{tikzpicture}[scale=0.5]
\draw[thick,dotted] (0,0) -- (0,1.5) node[midway, left] {\!\!{\footnotesize ${}_{0}$}};
\draw[black,fill=white] (0,1.5) circle (.75ex);
\end{tikzpicture}}}\quad,\qquad
\Pi\Big(~~\vcenter{\hbox{\begin{tikzpicture}[scale=0.5]
\draw[thick,densely dashed] (0,0) -- (0,1.5)  node[midway, left] {\!\!\!\!\!\!{\footnotesize ${}_{0}$}};
\end{tikzpicture}}}~~\Big)\,=~
\vcenter{\hbox{\begin{tikzpicture}[scale=0.5]
\draw[thick,densely dashed] (0,0) -- (0,1.5)  node[midway, left] {\!\!{\footnotesize ${}_{0}$}};
\draw[black,fill=white] (0,1.5) circle (.75ex);
\end{tikzpicture}}}\quad.
\end{flalign}
These are all the ingredients one needs to compute the quantum correlation functions.

\begin{ex}\label{ex:1pt}
As a simple illustration of this diagrammatic calculus, let us compute the $1$-point
function 
\begin{flalign}
\langle \varphi\rangle\,=\, \sum_{k=0}^\infty \Pi\big((\delta H)^k(\varphi)\big)\quad,\qquad \varphi\in\LL^0\quad,
\end{flalign}
for a generator in degree $0$ to the lowest non-trivial order in the 
formal parameter $\lambda$. The first step of this iterative process is computing
\begin{flalign}
\nn \delta H(\varphi)\,&=\,\delta\,\Big(~\vcenter{\hbox{\begin{tikzpicture}[scale=0.5]
\draw[thick] (0,0) -- (0,0.75);
\draw[thick,snake it] (0,0.75) -- (0,1.5);
\draw[black,fill=black] (0,0.75) circle (.75ex);
\end{tikzpicture}}}~\Big)~=~
\frac{\lambda}{2}~
\vcenter{\hbox{\begin{tikzpicture}[scale=0.5]
\draw[thick] (0,-0.75) -- (0,0);
\draw[thick,snake it] (0,0) -- (0,0.75);
\draw[black,fill=black] (0,0) circle (.75ex);
\draw[thick] (0,0.75) -- (-0.75,1.5);
\draw[thick] (0,0.75) -- (0.75,1.5);
\end{tikzpicture}}}~+~
\frac{\lambda^2}{3!}~
\vcenter{\hbox{\begin{tikzpicture}[scale=0.5]
\draw[thick] (0,-0.75) -- (0,0);
\draw[thick,snake it] (0,0) -- (0,0.75);
\draw[black,fill=black] (0,0) circle (.75ex);
\draw[thick] (0,0.75) -- (-0.75,1.5);
\draw[thick] (0,0.75) -- (0.75,1.5);
\draw[thick] (0,0.75) -- (0,1.5);
\end{tikzpicture}}}~+~
\lambda~
\vcenter{\hbox{\begin{tikzpicture}[scale=0.5]
\draw[thick] (0,-0.75) -- (0,0);
\draw[thick,snake it] (0,0) -- (0,0.75);
\draw[black,fill=black] (0,0) circle (.75ex);
\draw[thick,snake it] (0,0.75) -- (-0.75,1.5);
\draw[thick,densely dashed] (0,0.75) -- (0.75,1.5);
\end{tikzpicture}}}\\[4pt]
\,&=\,
\frac{\lambda}{2}~
\vcenter{\hbox{\begin{tikzpicture}[scale=0.5]
\draw[thick] (0,-0.75) -- (0,0);
\draw[thick,snake it] (0,0) -- (0,0.75);
\draw[black,fill=black] (0,0) circle (.75ex);
\draw[thick] (0,0.75) -- (-0.75,1.5);
\draw[thick] (0,0.75) -- (0.75,1.5);
\end{tikzpicture}}}~+~
\frac{\lambda^2}{3!}~
\vcenter{\hbox{\begin{tikzpicture}[scale=0.5]
\draw[thick] (0,-0.75) -- (0,0);
\draw[thick,snake it] (0,0) -- (0,0.75);
\draw[black,fill=black] (0,0) circle (.75ex);
\draw[thick] (0,0.75) -- (-0.75,1.5);
\draw[thick] (0,0.75) -- (0.75,1.5);
\draw[thick] (0,0.75) -- (0,1.5);
\end{tikzpicture}}}~+~
\lambda~
\vcenter{\hbox{\begin{tikzpicture}[scale=0.5]
\draw[thick] (0,-0.75) -- (0,0);
\draw[thick,snake it] (0,0) -- (0,0.75);
\draw[black,fill=black] (0,0) circle (.75ex);
\draw[thick,snake it] (0,0.75) -- (-0.75,1.5);
\draw[thick,densely dashed] (0,0.75) -- (0.75,1.5) node[midway, right] {{\footnotesize ${}_{\perp}$}};
\end{tikzpicture}}}~+~
\lambda~
\vcenter{\hbox{\begin{tikzpicture}[scale=0.5]
\draw[thick] (0,-0.75) -- (0,0);
\draw[thick,snake it] (0,0) -- (0,0.75);
\draw[black,fill=black] (0,0) circle (.75ex);
\draw[thick,snake it] (0,0.75) -- (-0.75,1.5);
\draw[thick,densely dashed] (0,0.75) -- (0.75,1.5) node[midway, right] {{\footnotesize ${}_{0}$}};
\end{tikzpicture}}}\quad,
\end{flalign}
where in the second line we have decomposed the dashed line according to \eqref{eqn:LLdecomposition01}.
Applying $\Pi$ to this expression gives $0$, hence we have to go to the next perturbative 
order to see quantum corrections. Using the algebraic property \eqref{eqn:bigH} of the homotopy $H$ on many lines,
we find
\begin{flalign}
\nn (\delta H)^2(\varphi)\,&=\,\delta H\Bigg(\frac{\lambda}{2}~
\vcenter{\hbox{\begin{tikzpicture}[scale=0.5]
\draw[thick] (0,-0.75) -- (0,0);
\draw[thick,snake it] (0,0) -- (0,0.75);
\draw[black,fill=black] (0,0) circle (.75ex);
\draw[thick] (0,0.75) -- (-0.75,1.5);
\draw[thick] (0,0.75) -- (0.75,1.5);
\end{tikzpicture}}}~+~
\frac{\lambda^2}{3!}~
\vcenter{\hbox{\begin{tikzpicture}[scale=0.5]
\draw[thick] (0,-0.75) -- (0,0);
\draw[thick,snake it] (0,0) -- (0,0.75);
\draw[black,fill=black] (0,0) circle (.75ex);
\draw[thick] (0,0.75) -- (-0.75,1.5);
\draw[thick] (0,0.75) -- (0.75,1.5);
\draw[thick] (0,0.75) -- (0,1.5);
\end{tikzpicture}}}~+~
\lambda~
\vcenter{\hbox{\begin{tikzpicture}[scale=0.5]
\draw[thick] (0,-0.75) -- (0,0);
\draw[thick,snake it] (0,0) -- (0,0.75);
\draw[black,fill=black] (0,0) circle (.75ex);
\draw[thick,snake it] (0,0.75) -- (-0.75,1.5);
\draw[thick,densely dashed] (0,0.75) -- (0.75,1.5) node[midway, right] {{\footnotesize ${}_{\perp}$}};
\end{tikzpicture}}}~+~
\lambda~
\vcenter{\hbox{\begin{tikzpicture}[scale=0.5]
\draw[thick] (0,-0.75) -- (0,0);
\draw[thick,snake it] (0,0) -- (0,0.75);
\draw[black,fill=black] (0,0) circle (.75ex);
\draw[thick,snake it] (0,0.75) -- (-0.75,1.5);
\draw[thick,densely dashed] (0,0.75) -- (0.75,1.5) node[midway, right] {{\footnotesize ${}_{0}$}};
\end{tikzpicture}}}\Bigg)\\[4pt]
\nn \,&=\, 
\delta\Bigg(\frac{\lambda}{2}~
\vcenter{\hbox{\begin{tikzpicture}[scale=0.5]
\draw[thick] (0,-0.75) -- (0,0);
\draw[thick,snake it] (0,0) -- (0,0.75);
\draw[black,fill=black] (0,0) circle (.75ex);
\draw[thick] (0,0.75) -- (-0.75,1.5);
\draw[thick] (0,0.75) -- (0.75,1.5);
\draw[thick,snake it] (-0.75,1.5) -- (-0.75,2.25);
\draw[black,fill=black] (-0.75,1.5) circle (.75ex);
\end{tikzpicture}}}~+~
\frac{\lambda^2}{3!}~
\vcenter{\hbox{\begin{tikzpicture}[scale=0.5]
\draw[thick] (0,-0.75) -- (0,0);
\draw[thick,snake it] (0,0) -- (0,0.75);
\draw[black,fill=black] (0,0) circle (.75ex);
\draw[thick] (0,0.75) -- (-0.75,1.5);
\draw[thick] (0,0.75) -- (0.75,1.5);
\draw[thick] (0,0.75) -- (0,1.5);
\draw[thick,snake it] (-0.75,1.5) -- (-0.75,2.25);
\draw[black,fill=black] (-0.75,1.5) circle (.75ex);
\end{tikzpicture}}}~+~
\frac{\lambda}{2}~
\vcenter{\hbox{\begin{tikzpicture}[scale=0.5]
\draw[thick] (0,-0.75) -- (0,0);
\draw[thick,snake it] (0,0) -- (0,0.75);
\draw[black,fill=black] (0,0) circle (.75ex);
\draw[thick,snake it] (0,0.75) -- (-0.75,1.5);
\draw[thick,densely dashed] (0,0.75) -- (0.75,1.5) node[midway, right] {{\footnotesize ${}_{\perp}$}};
\draw[thick,dotted] (-0.75,1.5) -- (-0.75,2.25) node[midway, left] {\!\!{\footnotesize ${}_{\perp}$}};
\draw[black,fill=black] (-0.75,1.5) circle (.75ex);
\end{tikzpicture}}}~-~
\frac{\lambda}{2}~
\vcenter{\hbox{\begin{tikzpicture}[scale=0.5]
\draw[thick] (0,-0.75) -- (0,0);
\draw[thick,snake it] (0,0) -- (0,0.75);
\draw[black,fill=black] (0,0) circle (.75ex);
\draw[thick,snake it] (0,0.75) -- (-0.75,1.5);
\draw[thick,densely dashed] (0,0.75) -- (0.75,1.5) node[midway, right] {{\footnotesize ${}_{\perp}$}};
\draw[thick] (0.75,1.5) -- (0.75,2.25)  ;
\draw[black,fill=black] (0.75,1.5) circle (.75ex);
\end{tikzpicture}}}
~+~
\lambda~
\vcenter{\hbox{\begin{tikzpicture}[scale=0.5]
\draw[thick] (0,-0.75) -- (0,0);
\draw[thick,snake it] (0,0) -- (0,0.75);
\draw[black,fill=black] (0,0) circle (.75ex);
\draw[thick,snake it] (0,0.75) -- (-0.75,1.5);
\draw[thick,densely dashed] (0,0.75) -- (0.75,1.5) node[midway, right] {{\footnotesize ${}_{0}$}};
\draw[thick, dotted] (-0.75,1.5) -- (-0.75,2.25) node[midway, left] {\!{\footnotesize ${}_{\perp}$}};
\draw[black,fill=black] (-0.75,1.5) circle (.75ex);
\end{tikzpicture}}}\Bigg)\\[4pt]
\,&=\, \frac{\hbar\, \lambda}{2}~
\vcenter{\hbox{\begin{tikzpicture}[scale=0.5]
\draw[thick] (0,-0.75) -- (0,0);
\draw[thick,snake it] (0,0) -- (0,0.75);
\draw[black,fill=black] (0,0) circle (.75ex);
\draw[thick] (0,0.75) -- (-0.75,1.5);
\draw[thick] (0,0.75) -- (0.75,1.5);
\draw[thick] (0.75,1.5) -- (0.75,2.25);
\draw[thick,snake it] (-0.75,1.5) -- (-0.75,2.25);
\draw[black,fill=black] (-0.75,1.5) circle (.75ex);
\draw[thick] (-0.75,2.25) to[out=90,in=90] (0.75,2.25);
\end{tikzpicture}}}~+~
\frac{\hbar\,\lambda}{2}~
\vcenter{\hbox{\begin{tikzpicture}[scale=0.5]
\draw[thick] (0,-0.75) -- (0,0);
\draw[thick,snake it] (0,0) -- (0,0.75);
\draw[black,fill=black] (0,0) circle (.75ex);
\draw[thick,snake it] (0,0.75) -- (-0.75,1.5);
\draw[thick,densely dashed] (0,0.75) -- (0.75,1.5) node[midway, right] {{\footnotesize ${}_{\perp}$}};
\draw[thick,densely dashed] (0.75,1.5) -- (0.75,2.25);
\draw[thick,densely dashed] (-0.75,2.25) to[out=90,in=90] (0.75,2.25);
\draw[thick,dotted] (-0.75,1.5) -- (-0.75,2.25) node[midway, left] {\!\!{\footnotesize ${}_{\perp}$}};
\draw[black,fill=black] (-0.75,1.5) circle (.75ex);
\end{tikzpicture}}}~-~
\frac{\hbar\,\lambda}{2}~
\vcenter{\hbox{\begin{tikzpicture}[scale=0.5]
\draw[thick] (0,-0.75) -- (0,0);
\draw[thick,snake it] (0,0) -- (0,0.75);
\draw[black,fill=black] (0,0) circle (.75ex);
\draw[thick,snake it] (0,0.75) -- (-0.75,1.5);
\draw[thick,snake it] (-0.75,1.5) -- (-0.75,2.25);
\draw[thick,densely dashed] (0,0.75) -- (0.75,1.5) node[midway, right] {{\footnotesize ${}_{\perp}$}};
\draw[thick] (0.75,1.5) -- (0.75,2.25)  ;
\draw[black,fill=black] (0.75,1.5) circle (.75ex);
\draw[thick] (-0.75,2.25) to[out=90,in=90] (0.75,2.25);
\end{tikzpicture}}}
~+~
\hbar\,\lambda~
\vcenter{\hbox{\begin{tikzpicture}[scale=0.5]
\draw[thick] (0,-0.75) -- (0,0);
\draw[thick,snake it] (0,0) -- (0,0.75);
\draw[black,fill=black] (0,0) circle (.75ex);
\draw[thick,snake it] (0,0.75) -- (-0.75,1.5);
\draw[thick,densely dashed] (0,0.75) -- (0.75,1.5) node[midway, right] {{\footnotesize ${}_{0}$}};
\draw[thick,densely dashed] (0.75,1.5) -- (0.75,2.25);
\draw[thick, dotted] (-0.75,1.5) -- (-0.75,2.25) node[midway, left] {\!{\footnotesize ${}_{\perp}$}};
\draw[black,fill=black] (-0.75,1.5) circle (.75ex);
\draw[thick,densely dashed] (-0.75,2.25) to[out=90,in=90] (0.75,2.25);
\end{tikzpicture}}}~+~\mathcal{O}(\lambda^2)\quad.
\end{flalign}
Note that the fourth term gives zero because the decomposition 
$\g = \g_0^{\perp} \oplus \g_{0}$ is orthogonal with respect to the Killing form.
By a similar identity as in \eqref{eqn:homotopypull}, namely
\begin{flalign}
	\vcenter{\hbox{\begin{tikzpicture}[scale=0.5]
		\draw[thick,snake it] (0,0) -- (0,0.75);
		\draw[thick,dotted] (0,0.75) -- (0,1.5) node[midway, left] {\!{\footnotesize ${}_{\perp}$}};
		\draw[black,fill=black] (0,0.75) circle (.75ex);
		\draw[thick, densely dashed] (0.75,0) -- (0.75,1.5) node[midway, right] {\!{\footnotesize ${}_{\perp}$}};
		\draw[thick, densely dashed] (0,1.5) to[out=90,in=90] (0.75,1.5);
	\end{tikzpicture}}}
	~=\, \Delta_{\BV}\big(h(\alpha)\, \chi\big)
	\,=\, \{h(\alpha),\chi\} \,=\, \{\alpha ,h(\chi)\}
	\,=\, -\Delta_{\BV}\big(\alpha\, h(\chi)\big) \,=\,- ~
	\vcenter{\hbox{\begin{tikzpicture}[scale=0.5]
		\draw[thick,snake it] (0,0) -- (0,1.5);
		\draw[thick] (0,1.5) to[out=90,in=90] (0.75,1.5);	
		\draw[thick,densely dashed] (0.75,0) -- (0.75,0.75) node[midway, right] {\!{\footnotesize ${}_{\perp}$}};
		\draw[thick] (0.75,0.75) -- (0.75,1.5);
		\draw[black,fill=black] (0.75,0.75) circle (.75ex);
	\end{tikzpicture}}}\quad,
\end{flalign}
one finds that the second and the third term coincide. 
Hence, the $1$-point function is given by
\begin{flalign}
\langle \varphi\rangle\,=\, \frac{\hbar\, \lambda}{2}~
\vcenter{\hbox{\begin{tikzpicture}[scale=0.5]
\draw[thick] (0,-0.75) -- (0,0);
\draw[thick,snake it] (0,0) -- (0,0.75);
\draw[black,fill=black] (0,0) circle (.75ex);
\draw[thick] (0,0.75) -- (-0.75,1.5);
\draw[thick] (0,0.75) -- (0.75,1.5);
\draw[thick] (0.75,1.5) -- (0.75,2.25);
\draw[thick,snake it] (-0.75,1.5) -- (-0.75,2.25);
\draw[black,fill=black] (-0.75,1.5) circle (.75ex);
\draw[thick] (-0.75,2.25) to[out=90,in=90] (0.75,2.25);
\end{tikzpicture}}}~+~
\hbar\,\lambda~
\vcenter{\hbox{\begin{tikzpicture}[scale=0.5]
\draw[thick] (0,-0.75) -- (0,0);
\draw[thick,snake it] (0,0) -- (0,0.75);
\draw[black,fill=black] (0,0) circle (.75ex);
\draw[thick,snake it] (0,0.75) -- (-0.75,1.5);
\draw[thick,densely dashed] (0,0.75) -- (0.75,1.5) node[midway, right] {{\footnotesize ${}_{\perp}$}};
\draw[thick,densely dashed] (0.75,1.5) -- (0.75,2.25);
\draw[thick,densely dashed] (-0.75,2.25) to[out=90,in=90] (0.75,2.25);
\draw[thick,dotted] (-0.75,1.5) -- (-0.75,2.25) node[midway, left] {\!\!{\footnotesize ${}_{\perp}$}};
\draw[black,fill=black] (-0.75,1.5) circle (.75ex);
\end{tikzpicture}}}~+~\mathcal{O}(\lambda^2)\quad.
\end{flalign}
Using the explicit form of the interaction part of the differential (see \eqref{eqn:totaldifferential}),
one can compute the numerical value of the ghost field contribution and one finds
\begin{flalign}\label{eqn:ghostloop}
\hbar\,\lambda~
\vcenter{\hbox{\begin{tikzpicture}[scale=0.5]
\draw[thick] (0,-0.75) -- (0,0);
\draw[thick,snake it] (0,0) -- (0,0.75);
\draw[black,fill=black] (0,0) circle (.75ex);
\draw[thick,snake it] (0,0.75) -- (-0.75,1.5);
\draw[thick,densely dashed] (0,0.75) -- (0.75,1.5) node[midway, right] {{\footnotesize ${}_{\perp}$}};
\draw[thick,densely dashed] (0.75,1.5) -- (0.75,2.25);
\draw[thick,densely dashed] (-0.75,2.25) to[out=90,in=90] (0.75,2.25);
\draw[thick,dotted] (-0.75,1.5) -- (-0.75,2.25) node[midway, left] {\!\!{\footnotesize ${}_{\perp}$}};
\draw[black,fill=black] (-0.75,1.5) circle (.75ex);
\end{tikzpicture}}}~=~\hbar\,\lambda \,\sum_{i=1}^{\dim(\g_0^\perp)}
2\,N\,\Tr_{\AA}\Big( t_i \, h\big([t_i,h(\varphi)]\big)\Big)\quad,
\end{flalign}
where $h$ are the components of the homotopy in \eqref{eqn:homotopy01}, $[\,\cdot\,,\,\cdot\,]$ is the Lie bracket
on $\g=\su(N)$ and $\{t_i\in \g_0^\perp\}$ is an orthogonal basis with respect to the Killing form,
i.e.\ $2N \,\Tr_{\AA}(t_i\,t_j) = - \delta_{ij}$ with the minus sign being a consequence of the Killing form 
being negative definite.
\sk

This ghost field contribution is in general non-zero, as one can simply show 
by setting $N=2$. As orthogonal basis for $\g=\su(2)$ we take the 
(appropriately normalized and anti-Hermitian) Pauli matrices
\begin{flalign}
t_i \,=\,\frac{\ii}{\sqrt{8}}\,\sigma_i~~,\quad\text{for }i=1,2,3\quad,
\end{flalign}
which satisfy the Lie bracket relations $[t_i,t_j] = -\frac{1}{\sqrt{2}}\,\epsilon_{ijk}\,t_k$.
The Lie sub-algebra $\g_0\subset \g$ is spanned by $t_3$
and its orthogonal complement $\g_0^\perp$ is spanned by $t_1$ and $t_2$.
Because the homotopy $h(\varphi)$ in \eqref{eqn:ghostloop} maps surjectively onto $\g_0$,
we can set without loss of generality $h(\varphi)=t_3$ by choosing an appropriate $\varphi$. 
The second homotopy in  \eqref{eqn:ghostloop} is of the type $h(\alpha)$ from \eqref{eqn:homotopy01}
and it acts on the basis of $\g_0^\perp$ as
\begin{flalign}
h(t_1)\,=\,\frac{1}{2\kappa}\,t_2~~,\quad
h(t_2)\,=\, -\frac{1}{2\kappa}\,t_1\quad.
\end{flalign}
This allows us to compute
\begin{flalign}
\eqref{eqn:ghostloop} \,\stackrel{N=2}{=}\, \hbar\,\lambda
\,4\,\Tr_{\AA}\Big(t_1\, h\big([t_1,t_3]\big) + t_2\,h\big([t_2,t_3]\big)\Big)\,=\, 
\frac{\hbar\,\lambda}{\sqrt{2}\,\kappa}\,=\,\hbar\,\lambda\,\sqrt{\frac{g_4}{3}}\quad,
\end{flalign}
hence the ghost field contribution does not vanish.
\end{ex}

\begin{ex}\label{ex:Dnontriv2pt}
With some more computational efforts, one can also
compute the $2$-point correlation function
\begin{flalign}
\langle \varphi_1\, \varphi_2\rangle \,=\,
\sum_{k=0}^\infty \Pi\big((\delta H)^k\big(\varphi_1\, \varphi_2\big)\big)
\quad, \quad\quad \varphi_{1},\varphi_{2} \in \LL^{0} = \DD^\vee\quad,
\end{flalign}
for two generators in degree $0$ to the 
lowest non-trivial order in the formal parameter $\lambda$. 
The final result is
\begin{flalign}
	\nn
	\langle \varphi_1\, \varphi_2\rangle \,&=~ \hbar~~
	\vcenter{\hbox{\begin{tikzpicture}[scale=0.5]
		\draw[thick] (0,0) -- (0,0.75);
		\draw[thick,snake it] (0,0.75) -- (0,1.5);
		\draw[black,fill=black] (0,0.75) circle (.75ex);
		\draw[thick] (0.75,0) -- (0.75,1.5);
		\draw[thick] (0,1.5) to[out=90,in=90] (0.75,1.5);
	\end{tikzpicture}}}
	~~+~~
	\hbar^{2}\lambda^{2}\,\Bigg(\frac{1}{2}~~
	\vcenter{\hbox{\begin{tikzpicture}[scale=0.5]
		\draw[thick] (0,0) -- (0,0.75);
		\draw[thick,snake it] (0,0.75) -- (0,1.5);
		\draw[black,fill=black] (0,0.75) circle (.75ex);
		\draw[thick] (0,1.5) -- (-0.75,2.25);
		\draw[black,fill=black] (-0.75,2.25) circle (.75ex);
		\draw[thick,snake it] (-0.75,2.25) -- (-0.75,3);
		\draw[thick] (-0.75,3) to[out=90,in=90] (0,3);
		\draw[thick] (0,1.5) -- (0,3);
		\draw[thick] (0,1.5) -- (0.75,2.25);
		\draw[thick] (1.5,0) -- (1.5,1.125);
		\draw[thick,snake it] (1.5,1.125) -- (1.5,2.25);
		\draw[black,fill=black] (1.5,1.125) circle (.75ex);
		\draw[thick] (0.75,2.25) to[out=90,in=90] (1.5,2.25);
	\end{tikzpicture}}}
	~~+~~
	\frac{1}{2}~~
	\vcenter{\hbox{\begin{tikzpicture}[scale=0.4]
		\draw[thick] (0,0) -- (0,0.75);
		\draw[thick,snake it] (0,0.75) -- (0,1.5);
		\draw[black,fill=black] (0,0.75) circle (.75ex);
		\draw[thick] (0,1.5) -- (-0.75,2.25);
		\draw[black,fill=black] (-0.75,2.25) circle (.75ex);
		\draw[thick,snake it] (-0.75,2.25) -- (-0.75,3);
		\draw[thick] (0,1.5) -- (0.75,2.25);
		\draw[thick] (-0.75,3) -- (-1.5,3.75);
		\draw[thick] (-0.75,3) -- (0,3.75);
		\draw[black,fill=black] (-1.5,3.75) circle (.75ex);
		\draw[thick,snake it] (-1.5,3.75) -- (-1.5,4.5);
		\draw[thick] (0,3.75) -- (0,4.5);
		\draw[thick] (-1.5,4.5) to[out=90,in=90] (0,4.5);
		\draw[thick] (1.5,0) -- (1.5,1.125);
		\draw[thick,snake it] (1.5,1.125) -- (1.5,2.25);
		\draw[black,fill=black] (1.5,1.125) circle (.75ex);
		\draw[thick] (0.75,2.25) to[out=90,in=90] (1.5,2.25);
	\end{tikzpicture}}}
	~~+~~
	\frac{1}{2}~~
	\vcenter{\hbox{\begin{tikzpicture}[scale=0.5]
		\draw[thick] (0,0) -- (0,0.75);
		\draw[thick,snake it] (0,0.75) -- (0,1.5);
		\draw[black,fill=black] (0,0.75) circle (.75ex);
		\draw[thick] (0,1.5) -- (-0.75,2.25);
		\draw[black,fill=black] (-0.75,2.25) circle (.75ex);
		\draw[thick,snake it] (-0.75,2.25) -- (-0.75,3);
		\draw[thick] (0,1.5) -- (0.75,2.25);
		\draw[black,fill=black] (0.75,2.25) circle (.75ex);
		\draw[thick,snake it] (0.75,2.25) -- (0.75,3);
		\draw[thick] (2.25,0) -- (2.25,0.75);
		\draw[thick,snake it] (2.25,0.75) -- (2.25,1.5);
		\draw[black,fill=black] (2.25,0.75) circle (.75ex);
		\draw[thick] (0.75,3) to[out=90,in=90] (1.5,3);
		\draw[thick] (2.25,1.5) -- (1.5,2.25);
		\draw[thick] (2.25,1.5) -- (3,2.25);
		\draw[thick] (1.5,2.25) -- (1.5,3);
		\draw[thick] (3,2.25) -- (3,3);
		\draw[thick] (-0.75,3) to[out=90,in=90] (3,3);
	\end{tikzpicture}}}\\
	\nn
	&\quad~~+~~
	\vcenter{\hbox{\begin{tikzpicture}[scale=0.4]
		\draw[thick] (0,0) -- (0,0.75);
		\draw[thick,snake it] (0,0.75) -- (0,1.5);
		\draw[black,fill=black] (0,0.75) circle (.75ex);
		\draw[thick] (0,1.5) -- (-0.75,2.25);
		\draw[black,fill=black] (-0.75,2.25) circle (.75ex);
		\draw[thick,snake it] (-0.75,2.25) -- (-0.75,3);
		\draw[thick] (0,1.5) -- (0.75,2.25);
		\draw[thick,snake it] (-0.75,3) -- (-1.5,3.75);
		\draw[thick,densely dashed] (-0.75,3) -- (0,3.75) node[midway, right] {{\footnotesize ${}_{\perp}$}};
		\draw[black,fill=black] (-1.5,3.75) circle (.75ex);
		\draw[thick,dotted] (-1.5,3.75) -- (-1.5,4.5) node[midway, left] {\!\!\!{\footnotesize ${}_{\perp}$}\!\!};
		\draw[thick,densely dashed] (0,3.75) -- (0,4.5);
		\draw[thick,densely dashed] (-1.5,4.5) to[out=90,in=90] (0,4.5);
		\draw[thick] (1.5,0) -- (1.5,1.125);
		\draw[thick,snake it] (1.5,1.125) -- (1.5,2.25);
		\draw[black,fill=black] (1.5,1.125) circle (.75ex);
		\draw[thick] (0.75,2.25) to[out=90,in=90] (1.5,2.25);
	\end{tikzpicture}}}
	~~+~~
	\vcenter{\hbox{\begin{tikzpicture}[scale=0.5]
		\draw[thick] (0,0) -- (0,0.75);
		\draw[thick,snake it] (0,0.75) -- (0,1.5);
		\draw[black,fill=black] (0,0.75) circle (.75ex);
		\draw[thick,snake it] (0,1.5) -- (-0.75,2.25);
		\draw[black,fill=black] (-0.75,2.25) circle (.75ex);
		\draw[thick,dotted] (-0.75,2.25) -- (-0.75,3) node[midway, left] {\!\!\!{\footnotesize ${}_{\perp}$}\!\!};
		\draw[thick,densely dashed] (0,1.5) -- (0.75,2.25) node[midway, right] {{\footnotesize ${}_{\perp}$}};
		\draw[black,fill=black] (0.75,2.25) circle (.75ex);
		\draw[thick] (0.75,2.25) -- (0.75,3);
		\draw[thick] (2.25,0) -- (2.25,0.75);
		\draw[thick,snake it] (2.25,0.75) -- (2.25,1.5);
		\draw[black,fill=black] (2.25,0.75) circle (.75ex);
		\draw[thick] (0.75,3) to[out=90,in=90] (1.5,3);
		\draw[thick,snake it] (2.25,1.5) -- (1.5,2.25);
		\draw[thick,densely dashed] (2.25,1.5) -- (3,2.25) node[midway, right] {{\footnotesize ${}_{\perp}$}};
		\draw[thick,snake it] (1.5,2.25) -- (1.5,3);
		\draw[thick,densely dashed] (3,2.25) -- (3,3);
		\draw[thick,densely dashed] (-0.75,3) to[out=90,in=90] (3,3);
	\end{tikzpicture}}}
	~~+~~
	\frac{1}{4}~~
	\vcenter{\hbox{\begin{tikzpicture}[scale=0.5]
		\draw[thick] (0,-0.75) -- (0,0);
		\draw[thick,snake it] (0,0) -- (0,0.75);
		\draw[black,fill=black] (0,0) circle (.75ex);
		\draw[thick] (0,0.75) -- (-0.75,1.5);
		\draw[thick] (0,0.75) -- (0.75,1.5);
		\draw[thick] (0.75,1.5) -- (0.75,2.25);
		\draw[thick,snake it] (-0.75,1.5) -- (-0.75,2.25);
		\draw[black,fill=black] (-0.75,1.5) circle (.75ex);
		\draw[thick] (-0.75,2.25) to[out=90,in=90] (0.75,2.25);
		\draw[thick] (2.25,-0.75) -- (2.25,0);
		\draw[thick,snake it] (2.25,0) -- (2.25,0.75);
		\draw[black,fill=black] (2.25,0) circle (.75ex);
		\draw[thick] (2.25,0.75) -- (1.5,1.5);
		\draw[thick] (2.25,0.75) -- (3,1.5);
		\draw[thick] (3,1.5) -- (3,2.25);
		\draw[thick,snake it] (1.5,1.5) -- (1.5,2.25);
		\draw[black,fill=black] (1.5,1.5) circle (.75ex);
		\draw[thick] (1.5,2.25) to[out=90,in=90] (3,2.25);
	\end{tikzpicture}}}
	\\
	&\quad~~+~~
	\frac{1}{2}~~
	\vcenter{\hbox{\begin{tikzpicture}[scale=0.5]
		\draw[thick] (0,-0.75) -- (0,0);
		\draw[thick,snake it] (0,0) -- (0,0.75);
		\draw[black,fill=black] (0,0) circle (.75ex);
		\draw[thick] (0,0.75) -- (-0.75,1.5);
		\draw[thick] (0,0.75) -- (0.75,1.5);
		\draw[thick] (0.75,1.5) -- (0.75,2.25);
		\draw[thick,snake it] (-0.75,1.5) -- (-0.75,2.25);
		\draw[black,fill=black] (-0.75,1.5) circle (.75ex);
		\draw[thick] (-0.75,2.25) to[out=90,in=90] (0.75,2.25);
		\draw[thick] (2.25,-0.75) -- (2.25,0);
		\draw[thick,snake it] (2.25,0) -- (2.25,0.75);
		\draw[black,fill=black] (2.25,0) circle (.75ex);
		\draw[thick,snake it] (2.25,0.75) -- (1.5,1.5);
		\draw[thick,densely dashed] (2.25,0.75) -- (3,1.5) node[midway, right] {{\footnotesize ${}_{\perp}$}};
		\draw[thick,densely dashed] (3,1.5) -- (3,2.25);
		\draw[thick,dotted] (1.5,1.5) -- (1.5,2.25) node[midway, left] {{\footnotesize ${}_{\perp}$}\!\!};
		\draw[black,fill=black] (1.5,1.5) circle (.75ex);
		\draw[thick,densely dashed] (1.5,2.25) to[out=90,in=90] (3,2.25);
	\end{tikzpicture}}}
	~~+~~
	\frac{1}{2}~~
	\vcenter{\hbox{\begin{tikzpicture}[scale=0.5]
		\draw[thick] (0,-0.75) -- (0,0);
		\draw[thick,snake it] (0,0) -- (0,0.75);
		\draw[black,fill=black] (0,0) circle (.75ex);
		\draw[thick,snake it] (0,0.75) -- (-0.75,1.5);
		\draw[thick,densely dashed] (0,0.75) -- (0.75,1.5) node[midway, right] {{\footnotesize ${}_{\perp}$}};
		\draw[thick,densely dashed] (0.75,1.5) -- (0.75,2.25);
		\draw[thick,dotted] (-0.75,1.5) -- (-0.75,2.25) node[midway, left] {\!\!\!{\footnotesize ${}_{\perp}$}\!\!};
		\draw[black,fill=black] (-0.75,1.5) circle (.75ex);
		\draw[thick,densely dashed] (-0.75,2.25) to[out=90,in=90] (0.75,2.25);
		\draw[thick] (2.25,-0.75) -- (2.25,0);
		\draw[thick,snake it] (2.25,0) -- (2.25,0.75);
		\draw[black,fill=black] (2.25,0) circle (.75ex);
		\draw[thick] (2.25,0.75) -- (1.5,1.5);
		\draw[thick] (2.25,0.75) -- (3,1.5);
		\draw[thick] (3,1.5) -- (3,2.25);
		\draw[thick,snake it] (1.5,1.5) -- (1.5,2.25);
		\draw[black,fill=black] (1.5,1.5) circle (.75ex);
		\draw[thick] (1.5,2.25) to[out=90,in=90] (3,2.25);
	\end{tikzpicture}}}
	~~+~~
	\vcenter{\hbox{\begin{tikzpicture}[scale=0.5]
		\draw[thick] (0,-0.75) -- (0,0);
		\draw[thick,snake it] (0,0) -- (0,0.75);
		\draw[black,fill=black] (0,0) circle (.75ex);
		\draw[thick,snake it] (0,0.75) -- (-0.75,1.5);
		\draw[thick,densely dashed] (0,0.75) -- (0.75,1.5) node[midway, right] {{\footnotesize ${}_{\perp}$}};
		\draw[thick,densely dashed] (0.75,1.5) -- (0.75,2.25);
		\draw[thick,dotted] (-0.75,1.5) -- (-0.75,2.25) node[midway, left] {\!\!\!{\footnotesize ${}_{\perp}$}\!\!};
		\draw[black,fill=black] (-0.75,1.5) circle (.75ex);
		\draw[thick,densely dashed] (-0.75,2.25) to[out=90,in=90] (0.75,2.25);
		\draw[thick] (2.25,-0.75) -- (2.25,0);
		\draw[thick,snake it] (2.25,0) -- (2.25,0.75);
		\draw[black,fill=black] (2.25,0) circle (.75ex);
		\draw[thick,snake it] (2.25,0.75) -- (1.5,1.5);
		\draw[thick,densely dashed] (2.25,0.75) -- (3,1.5) node[midway, right] {{\footnotesize ${}_{\perp}$}};
		\draw[thick,densely dashed] (3,1.5) -- (3,2.25);
		\draw[thick,dotted] (1.5,1.5) -- (1.5,2.25) node[midway, left] {{\footnotesize ${}_{\perp}$}\!\!};
		\draw[black,fill=black] (1.5,1.5) circle (.75ex);
		\draw[thick,densely dashed] (1.5,2.25) to[out=90,in=90] (3,2.25);
	\end{tikzpicture}}}\Bigg) ~+~ \mathcal{O}(\lambda^{3}) \quad.
\end{flalign}
As in the case of the $1$-point function from the previous Example \ref{ex:1pt}, 
we observe that there are non-trivial ghost field contributions that are not present
for perturbations around the trivial Dirac operator $D_0=0$, see also Example \ref{ex:Dis0}.
\end{ex}


\section*{Acknowledgments}
We would like to thank John Barrett for helpful discussions about
fuzzy spectral triples. We also would like to thank the anonymous referees
for useful comments that helped us to improve the paper.
H.N.\ is supported by a PhD Scholarship from the School of Mathematical 
Sciences of the University of Nottingham. 
A.S.\ gratefully acknowledges the financial support of 
the Royal Society (UK) through a Royal Society University 
Research Fellowship (URF\textbackslash R\textbackslash 211015)
and the Enhancement Awards (RGF\textbackslash EA\textbackslash 180270, 
RGF\textbackslash EA\textbackslash 201051 and RF\textbackslash ERE\textbackslash 210053).



\begin{thebibliography}{10}

\bibitem[AC09]{AschieriCastellani}
P.~Aschieri and L.~Castellani,
``Noncommutative $D=4$ gravity coupled to fermions,''
JHEP \textbf{06}, 086 (2009)
[arXiv:0902.3817 [hep-th]].


\bibitem[ADMW06]{AschieriMetric}
P.~Aschieri, M.~Dimitrijevic, F.~Meyer and J.~Wess,
``Noncommutative geometry and gravity,''
Class.\ Quant.\ Grav.\ \textbf{23}, 1883--1912 (2006)
[arXiv:hep-th/0510059 [hep-th]].


\bibitem[AK19]{Khalkhali1}
S.~Azarfar and M.~Khalkhali,
``Random finite noncommutative geometries and topological recursion,''
arXiv:1906.09362 [math-ph].


\bibitem[Bar15]{Barrett1}
J.~W.~Barrett,
``Matrix geometries and fuzzy spaces as finite spectral triples,''
J.\ Math.\ Phys.\ \textbf{56}, no.\ 8, 082301 (2015)
[arXiv:1502.05383 [math-ph]].


\bibitem[BDG19]{Barrett2}
J.~W.~Barrett, P.~Druce and L.~Glaser,
``Spectral estimators for finite non-commutative geometries,''
J.\ Phys.\ A \textbf{52}, no.\ 27, 275203 (2019)
[arXiv:1902.03590 [gr-qc]].


\bibitem[BG19]{BarrettGaunt}
J.~W.~Barrett and J.~Gaunt,
``Finite spectral triples for the fuzzy torus,''
arXiv:1908.06796 [math.QA].


\bibitem[BG16]{Barrett3}
J.~W.~Barrett and L.~Glaser,
``Monte Carlo simulations of random non-commutative geometries,''
J.\ Phys.\ A \textbf{49}, no.\ 24, 245001 (2016)
[arXiv:1510.01377 [gr-qc]].


\bibitem[BV81]{BV}
I.~A.~Batalin and G.~A.~Vilkovisky, 
``Gauge Algebra and Quantization,'' 
Phys.\ Lett.\ B {\bf 102}, 27--31 (1981).


\bibitem[BM20]{BeggsMajid}
E.~Beggs and S.~Majid, 
{\it Quantum Riemannian Geometry}, 
Grundlehren der mathematischen Wissenschaften {\bf 355}, 
Springer Verlag (2020).


\bibitem[BSS21]{BSSderived}
M.~Benini, P.~Safronov and A.~Schenkel,
``Classical BV formalism for group actions,''
{\it to appear in Communications in Contemporary Mathematics}
[arXiv:2104.14886 [math-ph]].


\bibitem[CC97]{SpectralAction2}
A.~H.~Chamseddine and A.~Connes,
``The Spectral action principle,''
Commun.\ Math.\ Phys.\ \textbf{186}, 731--750 (1997)
[arXiv:hep-th/9606001 [hep-th]].


\bibitem[Con94]{Connes}
A.~Connes,
{\it Noncommutative geometry}, 
Academic Press, Inc., San Diego, CA (1994).


\bibitem[Con96]{SpectralAction1}
A.~Connes,
``Gravity coupled with matter and foundation of noncommutative geometry,''
Commun.\ Math.\ Phys.\ \textbf{182}, 155--176 (1996)
[arXiv:hep-th/9603053 [hep-th]].


\bibitem[CG16]{CostelloGwilliam}
K.~Costello and O.~Gwilliam,
{\it Factorization Algebras in Quantum Field Theory: Volume 1},
Cambridge University Press (2016).


\bibitem[CG21]{CostelloGwilliam2}
K.~Costello and O.~Gwilliam,
{\it Factorization Algebras in Quantum Field Theory: Volume 2},
Cambridge University Press (2021).


\bibitem[Cra04]{Cra04}
M.~Crainic,
``On the perturbation lemma, and deformations,''
\href{https://arxiv.org/abs/math/0403266}{arXiv:math.AT/0403266}.


\bibitem[GGHZ22]{LargeN1}
G.~Ginot, O.~Gwilliam, A.~Hamilton and M.~Zeinalian,
``Large N phenomena and quantization of the Loday-Quillen-Tsygan theorem,''
Advances in Mathematics {\bf 409}, 108631 (2022)
[arXiv:2108.12109 [math.QA]].


\bibitem[Gwi12]{Gwilliam}
O.~Gwilliam,
{\it Factorization algebras and free field theories},
PhD thesis, Northwestern University (2012).
Available at \url{https://people.math.umass.edu/~gwilliam/thesis.pdf}.


\bibitem[GHZ22]{LargeN2}
O.~Gwilliam, A.~Hamilton and M.~Zeinalian,
``A homological approach to the Gaussian Unitary Ensemble,''
arXiv:2206.04256 [math-ph].


\bibitem[Ise19a]{Ise19a}
R.~A.~Iseppi,
``The BV formalism: Theory and application to a matrix model,''
Rev.\ Math.\ Phys.\ \textbf{31}, 1950035 (2019)
[arXiv:1610.03463 [math-ph]].


\bibitem[Ise19b]{Ise19b}
R.~A.~Iseppi,
``The BRST cohomology and a generalized Lie algebra cohomology: Analysis of a matrix model,''
arXiv:1909.05053 [math-ph].


\bibitem[IvS17]{IvS17}
R.~A.~Iseppi and W.~D.~van Suijlekom,
``Noncommutative geometry and the BV formalism: Application to a matrix model,''
J.\ Geom.\ Phys.\ \textbf{120}, 129--141 (2017)
[arXiv:1604.00046 [math-ph]].


\bibitem[HKP22]{Khalkhali2}
H.~Hessam, M.~Khalkhali and N.~Pagliaroli,
``Bootstrapping Dirac ensembles,''
J.\ Phys.\ A \textbf{55}, no.\ 33, 335204 (2022)
[arXiv:2107.10333 [hep-th]].


\bibitem[HKPV22]{KhalkhaliReview}
H.~Hessam, M.~Khalkhali, N.~Pagliaroli and L.~Verhoeven,
``From Noncommutative Geometry to Random Matrix Theory,''
{\em to appear in J.\ Phys.\ A}
[arXiv:2204.14216 [hep-th]].


\bibitem[KP21]{Khalkhali3}
M.~Khalkhali and N.~Pagliaroli,
``Phase transition in random noncommutative geometries,''
J.\ Phys.\ A \textbf{54}, no.\ 3, 035202 (2021)
[arXiv:2006.02891 [math-ph]].


\bibitem[KP22]{Khalkhali4}
M.~Khalkhali and N.~Pagliaroli,
``Spectral statistics of Dirac ensembles,''
J.\ Math.\ Phys.\ \textbf{63}, 053504 (2022)
[arXiv:2109.12741 [hep-th]].


\bibitem[Kra98]{Krajewski}
T.~Krajewski,
``Classification of finite spectral triples,''
J.\ Geom.\ Phys.\ \textbf{28}, 1--30 (1998)
[arXiv:hep-th/9701081].


\bibitem[NSS21]{NSSfuzzy}
H.~Nguyen, A.~Schenkel and R.~J.~Szabo,
``Batalin-Vilkovisky quantization of fuzzy field theories,''
Lett.\ Math.\ Phys.\ \textbf{111}, 149 (2021)
[arXiv:2107.02532 [hep-th]].


\bibitem[PS19]{Perez-Sanchez1}
C.~I.~Perez-Sanchez,
``Computing the spectral action for fuzzy geometries: 
from random noncommutative geometry to bi-tracial multimatrix models,''
arXiv:1912.13288 [math-ph].


\bibitem[PS21]{Perez-Sanchez2}
C.~I.~Perez-Sanchez,
``On multimatrix models motivated by random noncommutative geometry I: 
The functional renormalization group as a flow in the free algebra,''
Annales Henri Poincar{\'e} \textbf{22}, no.\ 9, 3095--3148 (2021)
[arXiv:2007.10914 [math-ph]].


\bibitem[PS22]{Perez-Sanchez3}
C.~I.~Perez-Sanchez,
``On multimatrix models motivated by random noncommutative geometry II: 
A Yang-Mills-Higgs matrix model,''
Annales Henri Poincar{\'e} \textbf{23}, no.\ 6, 1979--2023 (2022)
[arXiv:2105.01025 [math-ph]].


\bibitem[vS15]{vanSuijlekom}
W.~D.~van Suijlekom,
{\it Noncommutative geometry and particle physics},
Mathematical Physics Studies, Springer Verlag, Dordrecht (2015).


\end{thebibliography}
\end{document}